\documentclass[11pt,letterpaper]{article}

\usepackage{amsmath}                    
\usepackage{amssymb}                    
\usepackage{amsthm}                     
\usepackage{amsfonts}                   

\usepackage{caption}
\usepackage{subcaption}
\usepackage{enumerate}
\usepackage{array}
\usepackage[nospace,noadjust]{cite}

\usepackage{pgfplots}
\pgfplotsset{compat=newest}
\usepackage[margin=1in]{geometry}
\usepackage{bookmark}

\usepackage{titlesec}

\DeclareMathOperator*{\E}{{\rm E}}
\DeclareMathOperator*{\R}{\mathbb{R}}
\DeclareMathOperator*{\Z}{\mathbb{Z}}
\DeclareMathOperator*{\N}{\mathbb{N}}

\DeclareMathOperator*{\LL}{{\Lambda}}
\DeclareMathOperator*{\Ll}{{\lambda}}
\DeclareMathOperator*{\C}{{\cal C}}
\DeclareMathOperator*{\e}{{\rm e}}
\DeclareMathOperator*{\boundary}{{\partial\!}}

\DeclareMathOperator*{\bh}{{\hat{\beta}}}
\def\taumix{\tau_{\rm mix}}
\def\Tcoup{T_{\rm coup}}
\def\GD{{\cal M}}

\newcommand{\tightoverset}[2]{%
	\mathop{#2}\limits^{\vbox to -.5ex{\kern-1.37ex\hbox{$#1$}\vss}}}

\newcommand{\tightoversetsub}[2]{%
	\mathop{#2}\limits^{\vbox to -.5ex{\kern-1.6ex\hbox{$#1$}\vss}}}

\theoremstyle{definition}
\newtheorem{dfn}{Definition}[section]
\theoremstyle{plain}
\newtheorem{thm}[dfn]{Theorem}

\newtheorem{lem}[dfn]{Lemma}

\newtheorem{fact}[dfn]{Fact}

\newtheorem{Claim}[dfn]{Claim}

\begin{document}

\title{Random-Cluster Dynamics in $\Z^2$}
\author{Antonio Blanca\thanks{Computer Science Division, U.C. Berkeley, Berkeley, CA 94720. Email: {\tt ablanca@cs.berkeley.edu}. Research supported in part by an NSF Graduate Research Fellowship and NSF grant CCF-1420934.}
	\and Alistair Sinclair\thanks{Computer Science Division, U.C. Berkeley, Berkeley, CA 94720. Email: {\tt sinclair@cs.berkeley.edu}. Research supported in part by NSF grant CCF-1420934.}
}

\maketitle

\begin{abstract} 
	\noindent
	The random-cluster model has been widely studied as a unifying framework for random graphs,
	spin systems and electrical networks, but its dynamics have so far largely resisted analysis.
	In this paper we
	analyze the Glauber dynamics of the random-cluster model in the canonical case where the
	underlying graph is an $n \times n$ box in the Cartesian lattice~$\Z^2$.  Our main result is a
	$O(n^2\log n)$ upper bound for the mixing time at all values of the model parameter~$p$
	except the critical point $p=p_c(q)$, and for all values of the second model parameter~$q\ge 1$. We also provide a matching lower bound proving that our result is tight.
	Our analysis takes as its starting point the recent breakthrough by Beffara and
	Duminil-Copin on the location of the random-cluster phase transition in~$\Z^2$.
	It is reminiscent of similar results for spin systems such as the Ising and Potts models,
	but requires the reworking of several standard tools in the context of the random-cluster
	model, which is not a spin system in the usual sense.
\end{abstract}

	\section{Introduction}
	
	Let $G=(V,E)$ be a finite graph.  The {\it random-cluster model\/} on~$G$ with parameters
	$p\in(0,1)$ and $q>0$ assigns to each subgraph $(V,A \subseteq E)$ a probability
	\begin{equation}\label{eq:rcmeasure}
	\mu_{G,p,q}(A) \propto p^{|A|}(1-p)^{|E|-|A|} q^{c(A)},
	\end{equation}
	where $c(A)$ is the number of connected components in~$(V,A)$.  $A$~is a
	{\it configuration\/} of the model.
	
	The random-cluster model was introduced in the late 1960s by Fortuin and Kasteleyn~\cite{FK}
	as a unifying framework for studying random graphs, spin systems in physics and electrical networks; see the book~\cite{Grimmett} for extensive background. When $q=1$ this model corresponds
	to the standard bond percolation model on $G$ with parameter $p$,
	but when $q>1$ (resp., $q<1$) the probability measure favors subgraphs with more (resp., fewer) connected components, and is thus a strict generalization.
	
	For the special case of integer $q \ge 2$ the random-cluster model is, in a precise sense, dual to the classical ferromagnetic {\it $q$-state Potts model}, where configurations are assignments of spin values $\{1,\ldots,q\}$ to the vertices of~$G$; the duality is established via a coupling of the models (see, e.g., \cite{ES}). Consequently, the random-cluster model illuminates much of the physical theory of the Ising/Potts models. However, it should be stressed that the random-cluster model is not a ``spin system" in the usual
	sense: in particular, the probability that an edge~$e$ belongs to~$A$ does not depend only on
	the dispositions of its neighboring edges but on the entire configuration~$A$, since connectivity
	is a global property. 
	
	At the other extreme, when $q\rightarrow 0$, the set of (weak) limits that arise for various choices of $p$ contains fundamental distributions on $G$, including uniform measures over the spanning trees, spanning forests and connected subgraphs of $G$.
	
	\medskip\noindent
	{\bf The random-cluster model on $\Z^2$.}\ \ The random-cluster model is well defined for the infinite 2-dimensional lattice graph~$\Z^2$ as the limit of the sequence of random-cluster measures on $n\times n$ square regions $\LL_n$ of $\Z^2$ as $n$ goes to infinity. Recent breakthrough work of Beffara and Duminil-Copin~\cite{BDC} for the infinite measure has established the following {\it phase transition\/} at the critical value $p=p_c(q) = \sqrt{q}/(\sqrt{q}+1)$ for all $q \ge 1$: for $p < p_c(q)$ all connected components are finite with high probability\footnote{We say that an event occurs with high probability if it occurs with probability approaching $1$ as $n \rightarrow \infty$.}; while for $p > p_c(q)$ there is at least one infinite component with high probability. It was also established in \cite{BDC} that for $p < p_c(q)$ the model exhibits ``decay of connectivities", i.e., the probability that two vertices lie in the same connected component decays to zero exponentially with the distance between them. This property is analogous to the classical ``decay of correlations"
	that has long been known for the Ising model (see, e.g.,~\cite{Mart}). 
	
	In this paper, we explore the
	consequences of the Beffara-Duminil-Copin result for the {\it dynamics\/} of the model in the case most widely studied
	in the literature, when $G$ is an $n\times n$ square region~$\Lambda_n$ of~$\Z^2$ and $q\ge 1$.
	
	\medskip\noindent
	{\bf Glauber dynamics.}\ \ A {\it Glauber dynamics\/} for the random-cluster model is any local Markov
	chain on configurations that is reversible with respect to the measure~(\ref{eq:rcmeasure}),
	and hence converges to it. Specifically we will consider the ``heat-bath" dynamics,
	which at each
	step updates one edge of the current configuration~$A$ as follows:
	\begin{enumerate}[(i)]
		\setlength\itemsep{1em}
		\item pick an edge $e\in E$ uniformly at random (u.a.r.);
		\item replace $A$ by $A\cup \{e\}$ with probability \[\frac{\mu_{G,p,q}(A\cup\{e\})}{\mu_{G,p,q}(A\cup\{e\})+\mu_{G,p,q}(A\setminus\{e\})};\]
		\item else replace $A$ by $A \setminus \{e\}$.
	\end{enumerate}
	These transition probabilities can be easily computed, as explained in
	Section~\ref{section:prelim}.
	
	Glauber dynamics for spin systems have been widely studied in both statistical physics and
	computer science.  On the one hand, they provide a Markov chain Monte Carlo algorithm for
	sampling configurations of the system from the Gibbs distribution; on the other hand, they are a
	generally accepted model for the evolution of the underlying physical system.  The primary
	object of study is the {\it mixing time}, i.e., the number of steps until the dynamics is close to
	its stationary distribution, starting from any initial configuration.
	
	There has been much activity over the past two decades in analyzing Glauber dynamics for
	spin systems such as the Ising and Potts models, and deep connections have emerged
	between the mixing time and the phase structure of the physical model. In contrast, the Glauber
	dynamics for the random-cluster model remains very poorly understood.  The main reason
	for this appears to be the fact mentioned above that connectivity is a global property; this has
	led to the lack (until the recent breakthrough~\cite{BDC}) of a precise understanding of the
	phase transition, as well as the failure of existing Markov chain analysis tools.
	
	Essentially all existing bounds on the mixing time of the random-cluster Glauber dynamics (even in the
	simplest case of the mean-field model, where $G$ is the complete graph~\cite{BSmf}) are indirect,
	and proceed via a non-local dynamics (the so-called ``Swendsen-Wang" dynamics~\cite{SW}
	or its variants).  Comparison technology developed recently by
	Ullrich~\cite{Ullrich2,Ullrich3,Ullrich4}
	allows bounds for the Glauber dynamics of the Ising/Potts models to be translated to the
	Swendsen-Wang dynamics, and then again to the random-cluster dynamics.  This leads,
	for example, to an upper bound of $O(n^6\log^2 n)$ on the mixing time of the random-cluster
	dynamics in~$\Lambda_n\subset \Z^2$, at all values $p\ne p_c(q)$ for all integer $q\ge 2$.
	This approach has several serious limitations:
	\begin{enumerate}
		\setlength\itemsep{1em}
		\item The comparison method invokes linear algebra, and hence suffers an inherent penalty
		of at least $\Omega(n^4)$ in the mixing time bound\footnote{For a pair of functions $f,g:\N \rightarrow \R^+$, we say that $f(n) = O(g(n))$ (resp., $f(n) = \Omega(g(n))$) if there exists constants $c,n_0 > 0$ such that $f(n) \le c \cdot g(n)$ (resp., $f(n) \ge c \cdot g(n))$ for all $n > n_0$. We say that $f(n) = \Theta(g(n))$ when $f(n) = O(g(n))$ and $f(n) = \Omega(g(n))$.}; thus tight bounds can never be obtained in this way.
		\item The comparison method also yields no insight into the actual behavior of the random-cluster
		dynamics, so, e.g., it is unlikely to illuminate the connections with phase transitions.
		\item Since it relies on comparison with the Ising/Potts models, this analysis applies only for
		{\it integer\/} values of~$q$, while the random-cluster model is defined for all positive
		values of~$q$.
	\end{enumerate}
	
	\medskip\noindent
	{\bf Results.}\ \ In this paper we present the first direct analysis of the random-cluster dynamics
	and prove the following tight theorem for the important case of~$\Z^2$:
	\begin{thm}\label{thm:mainintro}
		For any $q\ge 1$, the mixing time of the Glauber dynamics for the random-cluster model on
		$\Lambda_n\subset \Z^2$ is $\Theta(n^2\log n)$ at all values of $p\ne p_c(q)$.
	\end{thm}
	\noindent
	Theorem~\ref{thm:mainintro}, as stated, holds for the random-cluster
	model with so-called ``free" boundary conditions (i.e., there are no edges in
	$\Z^2\setminus \Lambda_n$). In fact, as a consequence of our proof,
	it also holds for the case of ``wired" boundary
	conditions (in which all vertices on the external face of~$\Lambda_n$ are connected). 
	
	The main component of our result is the analysis of the sub-critical regime $p<p_c$; the
	result for the super-critical regime $p>p_c$ follows from it easily by self-duality of~$\Z^2$
	and the fact that $p_c$ is exactly the self-dual point~\cite{BDC}. Our sub-critical upper bound analysis
	makes crucial use of the exponential decay of connectivities for $p<p_c$ established
	recently by Beffara and Duminil-Copin~\cite{BDC}, as discussed earlier. This analysis
	is reminiscent of similar results for spin systems (such as the Ising model), in which
	exponential decay of {\it correlations\/} has been shown to imply rapid mixing~\cite{MOS}. However, since the
	random-cluster model exhibits decay of connectivities rather than decay of correlations,
	we need to rework the standard tools used in these contexts.  In particular, we make three innovations. 
	
	First, the classical notion of ``disagreement percolation"~\cite{JB},
	which is used to bound the speed at which influence can propagate in~$\Z^2$ under
	the dynamics, has to be non-trivially extended to take account of the fact that in the
	random-cluster model influence spreads not from vertex to vertex but from cluster
	to cluster.  Second, we need to translate the decay of connectivities in the infinite volume~$\Z^2$
	(as proved in~\cite{BDC}) to a stronger ``spatial mixing" property in finite volumes~$\Lambda_n$,
	with suitable {\it boundary conditions\/} around the external face; in doing this 
	we use the machinery developed by Alexander in~\cite{KA}, but adapted to hold for arbitrary (not just integer)~$q$
	and for a suitable class of boundary conditions that we call ``side-homogeneous"
	(see Section~\ref{section:prelim} for a definition).
	Finally, while we follow standard recursive arguments in relating the mixing
	time in~$\Lambda_n$ to that in smaller regions~$\Lambda_{n'}$ for $n'\ll n$, our approach
	differs in its sensitivity to the boundary conditions on the smaller regions: previous
	applications for spin systems have typically required rapid mixing to hold in $\Lambda_{n'}$ for
	{\it arbitrary\/} boundary conditions, while in our case we require it to hold only for
	side-homogeneous conditions. This aspect of our proof is actually essential because the
	random-cluster model does {\it not} exhibit spatial mixing for arbitrary boundary conditions (see Section \ref{section:boundary-conditions}); our definition of side-homogeneous conditions is
	motivated by the fact that they are both restricted enough to allow spatial mixing to hold,
	and general enough to make the recursion go through. 
	Our lower bound proof uses technology from analogous results for spin systems of Hayes and Sinclair \cite{HS}, again adapted to the random-cluster setting.
	
	Our results leave open the question of the mixing time at the critical point $p=p_c(q)$. 
	A full treatment of this regime may lie beyond the scope of
	current knowledge. In particular, the nature of the phase transition is not fully
	understood, but is conjectured to be second-order for $q \le 4$ and
	first-order for $q>4$;
	see \cite{LMMRS} and \cite{DCST} for partial results in this direction. This would suggest that the mixing time
	at~$p_c$ should be polynomial (though presumably of larger order than $O(n^2 \log n)$)
	for $q \le 4$ and exponential in~$n$ for $q>4$.  Indeed, the former already follows
	for $q=2$ from recent results of Lubetzky and Sly~\cite{LS} on the Ising model
	at criticality, and the latter (for periodic boundary conditions) for the case of sufficiently
	large integer~$q$ from results
	of Borgs et al.~\cite{BCT,BFKTVV} for the Glauber dynamics of the Potts model,
	using in both cases Ullrich's comparison techniques mentioned earlier.
	
	We also leave open whether Theorem~\ref{thm:mainintro} holds in more generality for arbitrary boundary conditions.
	In fact, it is an interesting open question whether spatial mixing is a necessary condition for fast mixing in the random-cluster setting.
	
	The rest of the paper is organized as follows. We conclude this Introduction with a brief
	discussion of related work.  Section~\ref{section:prelim} contains some basic terminology
	and facts used throughout the paper.  In Sections~\ref{section:propagation}
	and~\ref{section:boundary-conditions} we derive our main analytical tools: bounds on
	the rate of disagreement percolation, and strong mixing for side-homogeneous boundary
	conditions, respectively.  In Sections~\ref{section:mixing} and ~\ref{section:mixing-lower} we apply these tools to derive
	our mixing time bounds in the sub-critical regime $p<p_c(q)$. Finally, in Section~\ref{section:super-critical} we use planar duality to extend the result to the super-critical
	regime $p>p_c(q)$, thus completing the proof of Theorem~\ref{thm:mainintro}.

	\medskip\noindent
	{\bf Related work.}\ \ The random-cluster model has been the subject of extensive
	research in both the applied probability and statistical physics communities, which
	is summarized in the book by Grimmett~\cite{Grimmett}.  A central open
	problem was to rigorously establish the phase transition in~$\Z^2$
	at $p_c(q) = {\sqrt{q}}/({\sqrt{q}+1})$,
	though this was not achieved until 2012 by Beffara and Duminil-Copin~\cite{BDC}.
	Predating the location of the phase transition, Alexander~\cite{KA} showed that
	exponential decay of connectivities in~$\Z^2$
	(as was also established for $p<p_c(q)$ in~\cite{BDC})
	implies spatial mixing in~$\Lambda_n$ for a certain class of boundary
	conditions and integer~$q$.
	
	The Glauber dynamics for the Ising model on~$\Z^2$ is essentially completely understood
	thanks to decades of research. At all parameter values below the critical point~$\beta_c$,
	the mixing time in $\Lambda_n$ is $O(n^2\log n)$, while above~$\beta_c$ it is
	$\exp(\Omega(n))$ (see~\cite{Mart} for a comprehensive treatment, and also~\cite{LS} for the behavior at~$\beta_c$). 
	Analogous results for the $q$-state Potts model
	follow from the random-cluster results of Beffara and Duminil-Copin~\cite{BDC}
	and Alexander~\cite{KA1},
	combined with the earlier work of Martinelli, Olivieri and Schonmann~\cite{MOS}
	relating spatial mixing to mixing times.  Our work in the present paper
	carries through a parallel program for the random-cluster dynamics.
	
	Glauber dynamics for the random-cluster model on any graph are poorly understood.
	As explained earlier, almost all the known results are derived by translating mixing time
	bounds from the Ising/Potts models using comparison techniques due to
	Ullrich~\cite{Ullrich2,Ullrich3,Ullrich4}.  Such translations typically incur a substantial
	overhead due to the use of linear algebra, hold only for integer~$q$,
	and give no insight into the dynamics.
	One exception is a direct polynomial bound on the mixing time for the random-cluster dynamics
	on graphs with bounded tree-width due to Ge and \v{S}tefankovi\v{c}~\cite{GS}.
	
	Finally, we mention some results on {\it non-local\/} dynamics for the random-cluster
	model in the simpler {\it mean-field\/} case, where $G$ is the complete graph $K_n$.
	The classical (non-local) Swendsen-Wang dynamics for the Ising and Potts models~\cite{SW}
	may be viewed as a dynamics for the random-cluster model via the standard
	coupling~\cite{ES}, so recent tight mixing time results for this case \cite{GSV,LNNP}
	translate directly for integer values of~$q$.  The recent paper~\cite{BSmf} proves
	similar results for the related Chayes-Machta dynamics, which applies to any
	(not necessarily integer) $q>1$.
	
	\section{Preliminaries}
	\label{section:prelim}
	
	In this section we gather a number of definitions and background results that we will refer to repeatedly. More details and proofs can be found in the books \cite{Grimmett,LPW}.
	
	\bigskip\noindent
	{\bf Random-cluster model on $\Z^2$.}\ \ Let $\mathbb{L} = (\Z^2,\mathbb{E})$ be the square lattice graph, where for $u,v \in \Z^2$, $(u,v) \in \mathbb{E}$ iff $d(u,v) = 1$ with $d(\cdot,\cdot)$ denoting the Euclidean distance. Let $\LL_n \subseteq \Z^2$ be the set of vertices of $\mathbb{L}$ contained in a square box of side length $n$, and let $\LL=(\LL_n,E_n)$ be the graph whose edge set $E_n$ contains all edges in $\mathbb{E}$ with both endpoints in $\LL_n$. We use $\boundary\LL$ to denote the {\it boundary} of $\LL$; that is, the set of vertices in $\LL_n$ connected by an edge in $\mathbb{E}$ to $\LL_n^c = \Z^2 \setminus \LL_n$.
	
	A {\it random-cluster configuration} on $\LL$ corresponds to a subset $A$ of $E_n$. Alternatively, it is sometimes convenient to think of $A$ as a vector in $\{0,1\}^{|E_n|}$ indexed by the edges, where $A(e)=1$ iff $e \in A$. Edges belonging to $A$ are called {\it open}, and edges in $E_n \setminus A$ {\it closed}.
	
	For any random-cluster configuration $A^c$ on $\LL_n^c$, we may consider the conditional random-cluster measure induced in $\LL_n$ by $A^c$. To make this precise, we introduce the standard concept of {\it boundary conditions}. A boundary condition for $\LL$ is a partition $\eta = (P_1,P_2,...,P_k)$ of $\boundary \LL$ which encodes how the vertices of $\boundary \LL$ are connected in a fixed configuration $A^c$ on $\LL_n^c$; i.e., for all $u,v \in \boundary \LL$, $u,v \in P_i$ iff $u$ and $v$ are connected by a path in $A^c$ (see Figure \ref{prelim:figure1}(a)). In this case we also say that $u$ and $v$ are {\it wired} in $\eta$.
	
	For $A \subseteq E_n$ and a boundary condition $\eta$, let $c(A,\eta)$ be the number of connected components of $(\LL_n,A)$ when the connectivities from the boundary condition $\eta$ are also considered. More precisely,  if $C_1,C_2$ are connected components of $A$, and there exist $u \in C_1 \cap \boundary\LL$ and $v \in C_2 \cap \boundary\LL$ such that $u$ and $v$ are wired in $\eta$, then $C_1$ and $C_2$ are identified as the same connected component in $A$. The random-cluster measure on $\LL$ with boundary condition $\eta$ and parameters $p \in (0,1)$ and $q > 0$ is then given by
	\begin{equation}\label{prelim:rc-measure}
	\mu^{\eta}_{\LL,p,q} (A) = \frac{p^{|A|}(1-p)^{|E_n \setminus A|}q^{c(A,\eta)}}{Z^{\eta}_{\LL,p,q}},
	\end{equation}
	where $Z^{\eta}_{\LL,p,q}$ is the normalizing constant, or {\it partition function}. (Cf. equation (\ref{eq:rcmeasure}) in the Introduction, which corresponds to the special case when the boundary condition $\eta$ is ``free"; see below.) When $\LL$, $p$ and $q$ are clear from the context we will just write $\mu^\eta$.
	
	\begin{figure*}
		\begin{center}
			\includegraphics[page=1,clip,trim=68 623 250 74,scale=1.1]{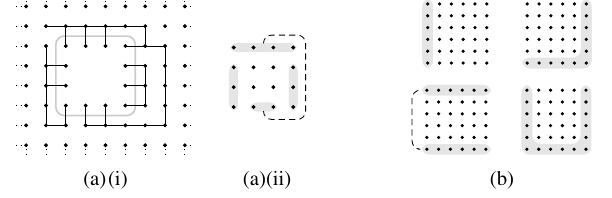}
			\caption{{\it {\rm (a): (i)} $\LL_4 \!\subset\! \Z^2$ with a random-cluster configuration $A^c$ in $\LL_4^c$, {\rm (ii)} the boundary condition induced in $\LL_4$ by $A^c$; {\rm (b)} examples of side-homogeneous boundary conditions.}}     
			\label{prelim:figure1}
		\end{center}
	\end{figure*}
	
	\bigskip\noindent
	{\bf Free, wired and side-homogeneous boundary conditions.}\ \ Some boundary conditions will be of particular interest to us. In the {\it free} boundary condition no two vertices of $\boundary \LL$ are wired. At the other extreme, in the {\it wired} boundary condition all vertices of $\boundary \LL$ are pairwise wired. We will use $\mu^0_{\LL,p,q}$ and $\mu^1_{\LL,p,q}$ to denote the random-cluster measures on $\LL$ with free and wired boundary conditions, respectively.
	
	We consider another class of boundary conditions which we call {\itshape side-homogeneous}. Let $L_1$, $L_2$, $L_3$, $L_4$ $\subset \boundary\LL$ be the sets of vertices on each {\it side} of the square box $\LL_n$. (A corner vertex of $\LL_n$ belongs to two sides.) The class of side-homogeneous boundary conditions contains all $\eta = (P_1,...,P_k)$ satisfying:
	\begin{enumerate}[(P1)]
		\setlength\itemsep{1em}
		\item $|P_i| > 1$ for at most one $i$; and
		\item  If $|P_i|> 1$, then $P_i$ is the union of some of the sets $L_j$; i.e., $P_i =\bigcup_{j \in \kappa} L_j$, for some $\kappa \subseteq \{1,2,3,4\}$.
	\end{enumerate}
	(See Figure \ref{prelim:figure1}(b).) Note that both the free and wired boundary conditions are side-homogeneous and there are in total 16 distinct side-homogeneous boundary conditions.
		
	\bigskip\noindent
	{\bf Monotonicity.}\ \ For any pair of boundary conditions $\eta$ and $\psi$, we say $\eta \le \psi$ if the partition $\eta$ is a refinement of $\psi$; i.e., if the connectivities induced by $\eta$ in $\boundary \LL$ are also induced by $\psi$.
	When $q \ge 1$, $\eta \le \psi$ implies $\mu^\eta_{\LL,p,q} \preceq \mu^\psi_{\LL,p,q}$, where $\preceq$ denotes {\it stochastic domination}; i.e.,
	$\mu^\eta_{\LL,p,q}({\cal E}) \le \mu^\psi_{\LL,p,q}({\cal E})$
	for all increasing events ${\cal E}$. (An event ${\cal E}$ is {\it increasing} if it is preserved by the addition of edges.)
	
	\bigskip\noindent
	{\bf Planar duality.}\ \ Let $\LL^*=(\LL_n^*,E_n^*)$ denote the planar dual of $\LL$ in the usual sense. That is, $\LL_n^*$ corresponds to the set of faces of $\LL$, and for each $e \in E_n$, there is a dual edge $e^* \in E_n^*$ connecting the two faces bordering $e$. The random-cluster measure satisfies
	$\mu_{\LL,p,q}(A) = \mu_{\LL^*,p^*,q}(A^*)$,
	where $A^*$ is the dual configuration to $A$ (i.e., $e^* \in E_n^*$ iff $e \in E_n$), and
	\[p^* = \frac{q(1-p)}{p+q(1-p)}.\]
	(This duality relation is a consequence of Euler's formula.) The unique value of $p$ satisfying $p=p^*$, denoted $p_{sd}(q)$, is called the {\it self-dual point}.
	
	\bigskip\noindent
	{\bf Infinite measure and phase transition.}\ \ The random-cluster measure may also be defined on the infinite lattice $\Z^2$ by considering the sequence of random-cluster measures on $\LL_n$ with free boundary conditions as $n \rightarrow \infty$. This sequence converges to a limiting measure $\mu_{\,\mathbb{L},p,q}$, which is known as the {\it random-cluster measure on $\mathbb{L}$}. The measure $\mu_{\,\mathbb{L},p,q}$ exhibits a phase transition corresponding to the appearance of an infinite connected component. That is, there exists a critical value $p = p_c(q)$ such that if $p < p_c(q)$ (resp., $p > p_c(q)$), then all components are finite (resp., there is at least one infinite component) with high probability.
	
	For $q \ge 1$, the exact value of $p_c(q)$ for $\mathbb{L}$ was only recently settled in breakthrough work by Beffara and Duminil-Copin \cite{BDC}, who proved the long standing conjecture
	\[p_c(q) = p_{sd} (q) = \frac{\sqrt{q}}{\sqrt{q}+1}.\]
	
	\bigskip\noindent
	{\bf Exponential decay of connectivies and spatial mixing.}\ \ In \cite{BDC}, it was also established that the phase transition is very sharp, meaning that as soon as $p < p_c(q)$ there is exponential decay of connectivities. More formally, for $q \ge 1$ and any fixed $p < p_c(q)$, there exist positive constants $C,\Ll$ such that for all $u,v \in \Z^2$,
	\begin{equation} \label{prelim:edc-infinite}
	\mu_{\,\mathbb{L},p,q}(u \leftrightarrow v) \le C {\e}^{-\Ll d(u,v)},\nonumber
	\end{equation}
	where $u \leftrightarrow v$ denotes the event that $u$ and $v$ are connected by a path of open edges. In work predating \cite{BDC}, Alexander \cite{KA} showed that exponential decay of connectivities implies exponential decay of {\it finite volume} connectivities uniformly over all boundary conditions. That is, for any boundary condition $\eta$ on $\LL$, and all $u,v \in \LL_n$,
	\begin{equation} \label{prelim:edc-finite}
	\mu^\eta_{\LL,p,q}(u \tightoverset{\scriptscriptstyle{\LL_n}}{\leftrightarrow} v) \le C {\rm e}^{-\Ll d(u,v)},
	\end{equation}
	where $u \tightoverset{\scriptscriptstyle{\LL_n}}{\leftrightarrow} v$ is the event that $u$ and $v$ are connected by a path of open edges in $\LL_n$.
	
	\begin{figure*}
		\begin{center}
			\includegraphics[page=1,clip,trim=90 627 250 60,scale=1.1]{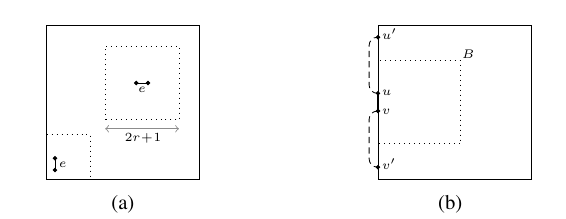}
			\caption{{\it {\rm (a)} $B(e,r)$ for two edges $e$ of $\LL$;
					{\rm (b)} a boundary condition $\psi$ where the spatial mixing property does not hold.}}     
			\label{prelim:figure2}
		\end{center}
	\end{figure*}
	
	The notion of decay of connectivities for the random-cluster model is analogous to the notion of decay of correlations in spin systems, which is ubiquitous in the spin systems literature. As in spin systems, we require in our analysis of the dynamics a stronger form of decay of connectivities known as {\it spatial mixing}.
	
	For $e \in E_n$, let $B(e,r) \subset \LL_n$ be the set of vertices in the minimal square box around $e$ such that $d(\{e\},v) \ge r$ for all $v \in \LL_n\setminus B(e,r)$. Note that if $d(\{e\},\boundary \LL) > r$, then $B(e,r)$ is just a square box of side length $2r+1$ centered at $e$; otherwise $B(e,r)$ intersects $\boundary \LL$ (see Figure \ref{prelim:figure2}(a)). Let $E(e,r)$ be the set of edges in $E_n$ with both endpoints in $B(e,r)$, and let $E^c(e,r) = E_n \setminus E(e,r)$. The spatial mixing property we use, which is slightly weaker than that defined in \cite{KA}, states that for all $e \in E_n$ and for every pair of configurations $A_1^c,A_2^c$ on $E^c = E^c(e,r)$, we have
	\begin{equation}\label{prelim:strong-mixing}
	\left|\mu^\eta_{\LL,p,q}(\,e=1\,|\,A_1^c\,)-\mu^\eta_{\LL,p,q}(\,e=1\,|\,A_2^c\,)\right| \le {\e}^{-\lambda r}
	\end{equation}
	for some constant $\lambda > 0$. 
	Alexander \cite{KA} showed that (\ref{prelim:edc-finite}) implies (\ref{prelim:strong-mixing}) for a certain class of boundary conditions $\eta$ when $q$ is an integer. In Section \ref{section:boundary-conditions} we will show, using the machinery developed in \cite{KA}, that (\ref{prelim:strong-mixing}) holds for all side-homogeneous boundary conditions $\eta$ for any (not necessarily integer) $q \ge 1$. 
	We shall see that (\ref{prelim:strong-mixing}) does not hold for arbitrary boundary conditions (see, e.g., Figure \ref{prelim:figure2}(b), together with the detailed explanation in Section \ref{section:boundary-conditions}).
		
	\bigskip\noindent
	{\bf Glauber dynamics.}\ \ The {\it Glauber dynamics}, which we denote $\GD$, is a local Markov chain on the random-cluster configurations of $\LL = (\LL_n,E_n)$ that is reversible with respect to $\mu^\eta_{\LL,p,q}$ for any boundary condition $\eta$; as a result, its stationary distribution is $\mu^\eta_{\LL,p,q}$. Given a random-cluster configuration $A_t \subseteq E_n$ at time $t$, a step of $\GD$ results in a new configuration $A_{t+1}$ as follows:
	\begin{enumerate}[(i)]
		\setlength\itemsep{1em}
		\item pick $e \in E_n$ u.a.r;
		\item let $A_{t+1} = A_t \cup \{e\}$ with probability
		\[\left\{\begin{array}{ll}
		\frac{p}{p+q(1-p)} & \mbox{if $e$ is a ``cut edge" in $({\LL}_n,A_t)$;} \\
		p & \mbox{otherwise;}
		\end{array}\right.\]
		\item else let $A_{t+1} = A_t \setminus \{e\}$.
	\end{enumerate}
	(We say $e$ is a {\it cut edge} in $(\LL_n,A_t)$ iff changing the current configuration of $e$ changes the number of connected components of $A_t$.) Note that this definition of the Glauber dynamics is equivalent to that in the Introduction, as can be seen by computing the ratios of the appropriate stationary weights $\mu^\eta_{\LL,p,q}(\cdot)$.
	
	\bigskip\noindent
	{\bf Mixing time and couplings.}\ \ Let $\Omega_{\textsc{\tiny RC}}$ be the set of random-cluster configurations of $\LL$, and let $\GD^t(X_0,\cdot)$ be the distribution of $\GD$ after $t$ steps starting from $X_0 \in \Omega_{\textsc{\tiny RC}}$. Let
	\begin{equation}\label{prelim:eq:mixing}
	\taumix(\varepsilon) := \max\limits_{X_0 \in \Omega_{\textsc{\tiny RC}}}\min\limits_t \left\{ ||\GD^t(X_0,\cdot)-\mu^\eta(\cdot)||_{\textsc{\tiny TV}} \le \varepsilon \right\}\nonumber,
	\end{equation}
	where $||\cdot||_{\textsc{\tiny TV}}$ denotes total variation distance. The {\it mixing time} of $\GD$ is given by $\taumix := \taumix(1/4)$. It is well-known that
	$\taumix(\varepsilon) \le \lceil \log_2 \varepsilon^{-1} \rceil\,\taumix$ for any positive $\varepsilon < 1/2$ 
	(see, e.g., \cite[Ch. 4.5]{LPW}).
	
	A {\it (one step) coupling} of the Markov chain $\GD$ specifies, for every pair of states $(X_t, Y_t) \in \Omega_{\textsc{\tiny RC}}^2$, a probability distribution over $(X_{t+1}, Y_{t+1})$ such that the processes $\{X_t\}$ and $\{Y_t\}$, viewed in isolation, are faithful copies of $\GD$, and if $X_t=Y_t$ then $X_{t+1}=Y_{t+1}$. The {\it coupling time}, denoted $\Tcoup$, is the minimum $T$ such that $\Pr[X_T \neq Y_T] \le 1/4$, starting from the worst possible pair of configurations $X_0$, $Y_0$. The following inequality is standard :
	$\taumix \le \Tcoup$ (see, e.g., \cite{LPW}).
	
	One coupling will be of particular interest to us. Namely, consider the coupling that couples the evolution of two copies of $\GD$, $\{X_t\}$ and $\{Y_t\}$, by using the same random $e \in E_n$ in step (i) and the same uniform random number $r \in [0,1]$ to decide whether to add or remove $e$ in steps (ii) and (iii). We call this the {\it identity coupling}. It is straightforward to verify that, when $q \ge 1$, the identity coupling is a {\it monotone coupling}, in the sense that if $X_t \subseteq Y_t$ then $X_{t+1} \subseteq Y_{t+1}$ with probability $1$. In fact, the identity coupling can be extended to a simultaneous coupling of {\it all} configurations that preserves the partial order $\subseteq$. Therefore, the coupling time starting from any pair of configurations is bounded by the coupling time for initial configurations $Y_0 = \emptyset$ and $X_0 = E_n$, which are the unique minimal and maximal elements in the partial order \cite{PW}.
	
	
	\section{The speed of disagreement percolation}
	\label{section:propagation}
	
	In spin systems, a central idea in the analysis of local Markov chains is to bound the speed at which information propagates. {\it Disagreement percolation} (or {\it path of disagreements}) arguments provide bounds of this sort  (see, e.g., \cite{JB,DSVW}). These arguments are based on the idea that in spin systems interactions only occur between neighboring sites, and thus if two configurations agree everywhere except in some region $A$, then it takes many steps 
	for a local Markov chain under the identity coupling
	to propagate these disagreements to regions that are far from $A$.
	
	In this section, we provide a bound on the speed of propagation of disagreements for the Glauber dynamics of the random-cluster model on $\LL = (\LL_n,E_n)$, under side-homogeneous boundary conditions. 
	A random-cluster configuration may exhibit long range interactions in the form of arbitrarily long paths, so disagreements could potentially propagate arbitrarily fast. Our insight is to restrict attention to pairs of configurations where one of them is {\it stationary}; i.e., it has law $\mu^\eta_{\LL,p,q}$. Then by (\ref{prelim:edc-finite}), the probability of long paths decays exponentially with the length of the path, which makes the long range interactions manageable.
	
	Throughout this section we will use the notation introduced in Section \ref{section:prelim}. 
	In addition, for a random-cluster configuration $A$ on $\LL$ and any $D \subseteq \LL_n$, we use $A(D)$ to denote the configuration induced by $A$ on the edges with both endpoints in $D$.
	Also, we use 
	$\boundary D$ and $\boundary E(D)$ to denote the set of vertices and edges, respectively, on the boundary of $D$; that is, $\boundary D$ is the set of vertices in $D$ connected by an edge in $E_n$ to $\LL_n \setminus D$ and
	$\boundary E(D) := \{(u,v) \in E_n: u \in D, v\not\in D\}$.
    (Note that if $v \in D \cap \boundary\LL$, then $v \not\in \boundary D$.) We are now ready to state and proof the main result of this section.
	
	\begin{lem} \label{lemma:propagation:dt}
		Let $p \!<\! p_c(q)$, $q \!\ge\! 1$ and consider two copies $\{X_t\}$, $\{Y_t\}$ of the Glauber dynamics on $\LL=(\LL_n,E_n)$ with a side-homogeneous boundary condition $\eta$. Assume $Y_0$ has law $\mu_{\LL,p,q}^\eta$ and that $X_0(B(e,r))=Y_0(B(e,r))$ for some $e\in E_n$ and $r \ge 1$. If the evolutions of $\{X_t\}$ and $\{Y_t\}$ are coupled using the identity coupling, then there exist absolute constants $c,C,\lambda > 0$ (independent of $r$ and $n$) such that, for $m = |E_n|$, $r \ge c$ and $1 \le k \le r^{1/4}/(4{\e}^2)$, we have
		\[\Pr[\,X_{km}(e) \neq Y_{km}(e)\,] ~\le~ C {\e}^{-\lambda r^{1/4}}.\]		
	\end{lem}
	
	\begin{proof} 
		Let  $B:=B(e,r)$. For some fixed $\ell$ (to be chosen later) and each $t \ge 0$ consider the event
		\begin{equation}
		\label{eq:only-small-cmpt}
		{\cal E}_{\ell,t} := \{u \tightoverset{\scriptscriptstyle{Y_t}}{\nleftrightarrow} v~~~\forall u,v \in B~\text{s.t.}~d(u,v)>\ell \},
		\end{equation}
		where $u\tightoverset{\scriptscriptstyle{Y_t}}{\nleftrightarrow}v$ denotes the event that $u$ and $v$ are not connected by a path in $Y_t(B)$. Let ${\cal E}_\ell \!:=\! \bigcap_{t=0}^{km} {\cal E}_{\ell,t}$; then,
		\begin{equation}\label{propagation:main-ineq}
		\Pr[\,X_{km}(e) \neq Y_{km}(e)\,] 
		\le \Pr[\,X_{km}(e) \neq Y_{km}(e) \,,\, {\cal E}_\ell\,] + \Pr[\,\neg {\cal E}_\ell\,].
		\end{equation}
		We bound each term on the right hand side of (\ref{propagation:main-ineq}) separately. 
		
		For any random-cluster configuration $A$ on $\LL$, let
		\begin{equation}\label{def:Gamma}
		\Gamma(A,B) := B \setminus  \bigcup\nolimits_{v \in \boundary B}\, C(v,A),
		\end{equation}
		where $C(v,A)$ is the set of vertices in the connected component of $v$ in $(\LL_n,A)$.
		
		\begin{figure*}
			\begin{center}
				\includegraphics[page=1,clip,trim=95 627 250 55,scale=1.1]{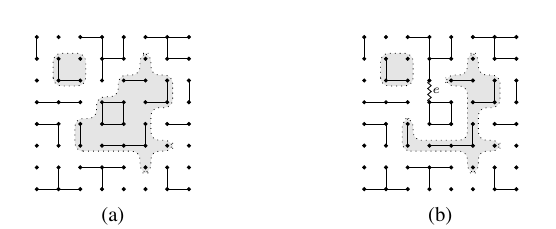}
				\caption{{\it {\rm (a)} a random-cluster configuration $A$ in $B$, with $R_0 = \Gamma(A,B)$ (shaded); {\rm (b)} $R_1$, assuming edge $e$ was updated in the first step.}}
				\label{propagation:figure2}
			\end{center}
		\end{figure*}
		
		Consider the sequence $R_0 \supseteq R_1 \supseteq ...$ of subsets (or {\it regions}) of $B$, such that $R_0 = \Gamma(Y_0,B)$ and
		\begin{equation}
		R_{t+1} = \begin{cases}
		R_t & \text{if no edge from $\boundary E(R_t)$ is updated at time $t$;}\\
		\\
		R_t \setminus C(a_t,Y_t) & \parbox[t]{.6\textwidth}{if $(a_t,b_t) \in \boundary E(R_t)$ with $a_t\in R_t$, $b_{t} \not\in R_t$ is the edge updated at time $t$;}
		\end{cases} \nonumber
		\end{equation}		
		(see Figure~\ref{propagation:figure2}). 
		The second case above applies 
		regardless of the state of $(a_t,b_t)$ in $X_{t+1}$ and $Y_{t+1}$.
		Observe that $R_t$ need not be a connected region of $\Z^2$
		and that every edge in $\boundary E(R_{t})$ is closed in both $X_t$ and $Y_t$. 
		
		The key observation is that, for all $t \ge 0$, $R_t$ is a region of $B$ in which $X_t(R_t) = Y_t(R_t)$. 
		We prove this by induction on $t$.		
		Assume $X_t(R_t) = Y_t(R_t)$ and let $\psi_{\textsc{X}}$ (resp., $\psi_{\textsc{Y}}$) be the boundary condition induced in $R_t$ by $X_t(\LL_n\setminus R_t)$ (resp., $Y_t(\LL_n\setminus R_t)$) and $\eta$. 
		Both $\psi_{\textsc{X}}$ and $\psi_{\textsc{Y}}$ are partitions of $\boundary R_t \cup (R_t \cap \boundary \LL)$. 
		(Recall that, by definition, if $v \in R_t \cap \boundary \LL$ then $v \not\in \boundary R_t$.) 
		
		We consider three cases based on the location of the edge $(a_t,b_t)$ with respect to $R_t$. First, if $\{a_t,b_t\} \cap R_t = \emptyset$, then clearly $X_{t+1}(R_{t+1}) = Y_{t+1}(R_{t+1})$. The second possibility is that $\{a_t,b_t\} \cap R_t = \{a_t\}$; if this is the case, then $R_{t+1} = R_t \setminus C(a_t,Y_t)$ by definition.
		Moreover,  $C(a_t,X_t) = C(a_t,Y_t)$
		because $X_t(R_t) = Y_t(R_t)$ and 
		all edges in $\boundary E(R_{t})$ are closed in both $X_t$ and $Y_t$;
		as a result $X_{t+1}(R_{t+1}) = Y_{t+1}(R_{t+1})$.
		
		Finally, when $\{a_t,b_t\} \subseteq R_t$ we show that $\psi_{\textsc{X}} = \psi_{\textsc{Y}}$, from which it follows that $X_{t+1}(R_{t+1}) = Y_{t+1}(R_{t+1})$.		
		First observe that every vertex in $\boundary R_t$ is a singleton in both $\psi_{\textsc{X}}$ and $\psi_{\textsc{Y}}$, since every edge in $\boundary E(R_{t})$ is closed in $X_t$ and $Y_t$.
		Therefore, 
		if $|R_t \cap \boundary \LL| \le 1$, then $\psi_{\textsc{X}}$, $\psi_{\textsc{Y}}$ are both the free boundary condition on $R_t$ and we are done.				
		Otherwise, assume that $|R_t \cap \boundary \LL| \ge 2$ and
		let $u$, $v$ be any two distinct vertices in $R_t \cap \boundary \LL$.
		If $u$ and $v$ are wired in $\eta$, then they are also wired in $\psi_{\textsc{X}}$ and $\psi_{\textsc{Y}}$. Moreover, if $u$ and $v$ are not wired in $\eta$, then property (P1) of side-homogeneous boundary conditions implies that one of them (say, $v$) is necessarily a singleton element in $\eta$. 
		Since there is no path of open edges from $v$ to $\LL_n\setminus R_t$ in either $X_t$ or $Y_t$, then $v$ is also a singleton element (and thus not wired to $u$) in both $\psi_{\textsc{X}}$ and $\psi_{\textsc{Y}}$.
		Hence, $u, v \in R_t \cap \boundary \LL$ are wired in $\psi_{\textsc{X}}$ (resp., $\psi_{\textsc{Y}}$) iff they are wired in $\eta$. 
		Since all the vertices in $\boundary R_t$ are singletons in both $\psi_{\textsc{X}}$ and $\psi_{\textsc{Y}}$, we conclude that $\psi_{\textsc{X}}=\psi_{\textsc{Y}}$.
		
		We now have that $X_{t}(R_{t}) = Y_{t}(R_{t})$ for all $t \ge 0$. Hence, 
		if both endpoints of $e$ lie in $R_{km}$, then $X_{km}(e)= Y_{km}(e)$.
		Also, we will choose $\ell \ll r$, 
		so when ${\cal E}_\ell$ occurs,		
		both endpoints of $e$ lie in $R_0$. 
		So, if $X_{km}(e) \neq Y_{km}(e)$, we may take $v_0$ to be the first endpoint of $e$ to be removed from $R_{t_0}$, at some time $t_0 \le km$. Let $e_1 \in \boundary E(R_{t_0})$ be the edge whose update is responsible for removing $v_0$ from $R_{t_0}$. Starting from $e_1$ we can then construct a sequence of edges $e_1,e_2,...$ such that $e_i=(u_i,v_i) \in \boundary  R_{t_{i-1}}$, with $u_i \in R_{t_{i-1}}$ and $v_i \not\in R_{t_{i-1}}$, is the edge that removes $v_{i-1}$ from $R_{t_{i-1}}$ at time $t_{i-1}$. Note that $t_0 > t_1 > ...$ and that the sequence $e_1,e_2,...$ stops once it reaches an edge $e_{\tilde{t}}$ incident to a vertex outside $R_0$. We call the sequence $e_1,e_2,...,e_{\tilde{t}}$ a {\it witness} for the fact that $X_{km}(e) \neq Y_{km}(e)$.
		
		The vertices $v_{i-1}$ and $u_{i}$ are in the same connected component in $Y_{t_{i-1}}$, and on the event ${\cal E}_\ell$, we have $d(v_{i-1},u_{i}) \le \ell$. Therefore, the number of witnesses of length $L$ is (crudely) at most $(4 (\ell+1)^2)^L$. Note also that every witness contains at least $r/(\ell + 1)$ edges, otherwise it cannot reach any of the vertices outside $R_0$. Moreover, the probability that a given witness of length $L$ is updated by the identity coupling in $km$ steps is $\binom{km}{L}\left(\frac{1}{m}\right)^L$. Hence,
		\begin{align}
		\Pr[\,X_{km}(e) \neq Y_{km}(e) \,,\, {\cal E}_\ell\,] 
		&\le~~ \sum_{L \ge \frac{r}{\ell+1}} {\binom{km}{L}\left(\frac{1}{m}\right)^L} (4 (\ell+1)^2)^L \nonumber\\
		&\le~~ \sum_{L \ge \frac{r}{\ell+1}} \left(\frac{4 {\e} k (\ell+1)^2}{L}\right)^{L}\nonumber\\
		&\le~~ \omega^{ \frac{r}{\ell+1}} \sum_{L \ge 0} \omega^{L},\nonumber
		\end{align}
		where $\omega = \frac{4 {\e} k (\ell+1)^3}{r}$. By taking $\ell = r^{1/4}-1$ and using the fact that $k \le r^{1/4}/(4 {\e}^2)$, we have
		\begin{equation}\label{propagation:first-term}
		\Pr[\,X_{km}(e) \neq Y_{km}(e) \,,\, {\cal E}_\ell\,] \le \frac{\e}{\e-1} \cdot {\e}^{-r^{3/4}}.
		\end{equation}
		
		Now we turn our attention to the second term on the right hand side of (\ref{propagation:main-ineq}). Let $N$ be the number of updates the identity coupling performs in $B$ up to time $km$, and let $M$ be the number of edges in $B$; i.e., $M := |E(e,r)| = \Theta(r^2)$. A Chernoff bound implies that $N > 2 k M$ with probability $\exp(-\Omega(k M))$, and thus
		\[\Pr[\,\neg{\cal E}_\ell\,] \le \Pr[\,\neg{\cal E}_\ell\,,\,N \le 2 k M\,]+{\e}^{-\Omega(k M)}.\]
		Observe that $\neg{\cal E}_\ell = \bigcup_{t=0}^{km} \neg{\cal E}_{\ell,t}$, and if the edge update at time $t$ occurs outside $B$, we have $\neg{\cal E}_{\ell,t} = \neg{\cal E}_{\ell,t+1}$. Hence, a union bound implies
		\begin{equation}\label{eq:dis-union-bound}
		\Pr[\,\neg{\cal E}_\ell\,] ~\le~ 2 k M \max_{0 \le t \le km}  \Pr[\,\neg{\cal E}_{\ell,t}\,]+{\e}^{-\Omega(k M)}.
		\end{equation}
		To bound $\Pr[\,\neg{\cal E}_{\ell,t}\,]$ we use the exponential decay of finite volume connectivities (\ref{prelim:edc-finite}). To do this, recall that $Y_0$ (and thus $Y_t$ for all $t \ge 0$) has law $\mu^\eta$. Also, there are only $O(r^4)$ pairs of vertices in $B$. Hence, (\ref{prelim:edc-finite}) and a union bound imply \[\Pr[\,\neg{\cal E}_{\ell,t}\,] \le O(r^4) \cdot {\e}^{-\Omega(\ell)} \le O(r^4) \cdot {\e}^{-\Omega(r^{1/4})}.\]
		Since $1 \le k \le r^{1/4}/(4{\e}^2)$ and $M = \Theta(r^2)$, (\ref{eq:dis-union-bound}) gives
		\begin{equation}
		\label{eq:neg:small-cmpt}
		\Pr[\,\neg{\cal E}_\ell\,] \le O(r^{6.25}) \cdot {\e}^{-\Omega(r^{1/4})} + {\e}^{-\Omega(r^{2})}.
		\end{equation}
		Together with (\ref{propagation:main-ineq}) and (\ref{propagation:first-term}), this implies that there exist constants $c,C,\Ll > 0$ such that for all $r \ge c$, we have 
		\[\Pr[\,X_{km}(e) \neq Y_{km}(e) \,] \le C {\e}^{-\lambda r^{1/4}},\]
		as desired.
	\end{proof}

	\section{Spatial mixing for side-homogeneous boundary conditions} \label{section:boundary-conditions}
	
	In this section we show that the spatial mixing property (\ref{prelim:strong-mixing}) holds for the class of side-homogeneous boundary conditions on $\LL=(\LL_n,E_n)$. Let $e\in E_n$ and let $B\!=\!B(e,r)$ for some $r \ge 1$. Spatial mixing holds when the influence on $e$ of the configuration in $E^c \!=\! E^c(e,r)$ decays exponentially with $r$. This is easy to establish when $B \cap \boundary\LL = \emptyset$, since such influence is present only if there are paths of open edges from $e$ to $\boundary B$, and, by (\ref{prelim:edc-finite}), the probability of such paths decays exponentially with $r$.  However, if $B \cap \boundary\LL \neq \emptyset$, the influence from $E^c$ could also propagate along $B \cap \boundary \LL$ via the boundary condition on $\LL$. This is why (\ref{prelim:strong-mixing}) does not hold for arbitrary boundary conditions, as the following concrete example illustrates. 
	
	With a slight abuse of notation, we use $\{E^c=1\}$ (resp., $\{E^c=0\}$) to denote the event that all the edges in $E^c$ are open (resp., closed). Suppose $e\!=\!(u,v)$ is an edge in $\boundary\LL$ that is far from the corners of $\LL$, and let $\psi$ be the boundary condition on $\LL$ where $u$ is wired to a vertex $u^\prime \in \boundary \LL \setminus B$ and $v$ is wired to a different vertex $v^\prime \in \boundary \LL \setminus B$ (see Figure \ref{prelim:figure2}(b)). When $p=1/2$ and $q=3$, we have $\mu^\psi(e=1|E^c=1) = 1/2$. Also, by considering a small box around $e$, is easy to check that $\mu^\psi(e=1|E^c=0) \le 2/5$. Both these bounds are independent of $r$ and $n$; consequently, $\psi$ does not have the spatial mixing property. 
	
	It turns out that side-homogeneous boundary conditions (and in particular property (P1)) rule out the possibility of influence propagating along $B \cap \boundary \LL$. As a result, we are able to establish the spatial mixing property for side-homogeneous boundary conditions, as stated in the following lemma; the proof uses the machinery developed in \cite{KA}.
	
	\begin{lem}\label{lemma:bc-smp} Let $p < p_c(q)$, $q \!\ge\! 1$ and let $\eta$ be a side-homogeneous boundary condition for $\LL = (\LL_n,E_n)$. For any $e \in E_n$, there exist constants $c,\lambda\!>\!0$ such that for all $r \ge c$ and every pair of configurations $A_1^c$, $A_2^c$ on $E^c$:
		\begin{equation}\label{bc:smp}
		\left|\mu^\eta(\,e=1\,|\,A_1^c\,)-\mu^\eta(\,e=1\,|\,A_2^c\,)\right| ~\le~ {\e}^{-\lambda r}.
		\end{equation}
	\end{lem}
	
	\begin{proof} Consider the measure $\mu^w := \mu^{\eta}(\,\cdot\,|E^c=1)$ on $B=B(e,r)$. Let $\Gamma(\cdot,\cdot)$ be defined as in (\ref{def:Gamma}). We derive the result from the following key fact, which we prove later.
		
		\begin{Claim}\label{claim:shbc} There exists a coupling $\pi$ of the distributions $\mu^\eta(\,\cdot\,|\,A_1^c\,)$, $\mu^\eta(\,\cdot\,|\,A_2^c\,)$ and $\mu^w$ such that $\pi(A_1,A_2,A_w) > 0$ only if $A_1 \subseteq A_w$, $A_2 \subseteq A_w$ and $A_1$, $A_2$ agree on all edges with both endpoints in $\Gamma(A_w,B)$.
		\end{Claim}
		
		\noindent Let $\Gamma^c(A_w,B) := \LL_n \setminus \Gamma(A_w,B)$. Given the coupling $\pi$, we have
		\begin{align}
		|\mu^\eta(\,e=1\,|\,A_1^c\,)-\mu^\eta(\,e=1 \,|\,A_2^c\,)| 
		& \le \pi(\,A_1(e) \neq A_2(e)\,) \nonumber\\
		& \le \pi\left(\,\Gamma^c(A_w,B) \cap \{e\} \neq \emptyset \,\right) \nonumber\\
		& \le \mu^w \left(\, \{e\} \leftrightarrow \boundary B  \,\right), \nonumber
		\end{align}
		where  $\{e\} \leftrightarrow \boundary B$ denotes the event that there is a path from $e$ to $\boundary B$. 
		
		Now, exponential decay of finite volume connectivities (\ref{prelim:edc-finite}) and a union bound over the boundary vertices imply
		\begin{equation}
		\mu^w (\, \{e\} \leftrightarrow \boundary B  \,) \le 2 C |\boundary B| {\e}^{-\Ll r}.
		\nonumber
		\end{equation}
		Since $|\boundary B| = \Theta(r)$, we obtain (\ref{bc:smp}) for all $ r \ge c$ for some constant $c>0$, and hence the lemma. \end{proof}
	
	\noindent
	We conclude this section by providing the missing proof of Claim \ref{claim:shbc}.
	
	\begin{proof} Let $\theta_1$ (resp., $\theta_2$) be the boundary condition induced on $B=B(e,r)$ by $A_1^c$ (resp., $A_2^c$) and $\eta$. Note that $\mu^{\theta_1}$, $\mu^{\theta_2}$ and $\mu^w$ are random-cluster measures on $B$ with different boundary conditions, and clearly  $\mu^w \succeq \mu^{\theta_1}$ and $\mu^w \succeq \mu^{\theta_2}$. Strassen's theorem (see, e.g., \cite{ST}) then implies the existence of monotone couplings $\mu_1$ for $\mu^{\theta_1}$ and $\mu^w$, and $\mu_2$ for $\mu^{\theta_2}$ and $\mu^w$. 
		(Recall that $\mu_1$ is a monotone coupling for $\mu^{\theta_1}$ and $\mu^w$ iff every sample $(A_{\theta_1},A_w)$ from $\mu_1$ satisfies $A_{\theta_1} \subseteq A_w$.) 
		We show next how to use $\mu_1$ and $\mu_2$ to construct the desired coupling $\pi$.
		
		First, let $\Delta := \Gamma(A_w,B)$ and let $\xi$ be the boundary condition induced in $\Delta$ by $\eta$ and the configuration of $A_w$ in $\Gamma^c(A_w,B)$.
		We construct $\pi$ as follows:
		\begin{enumerate}[(i)]
			\setlength\itemsep{1em}
			\item sample $(A_{\theta_1},A_w)$ from $\mu_1$;
			\item sample $A_{\theta_2}$ from $\mu_2(\,\cdot\,|\,A_w\,)$; and
			\item sample $A_\gamma$ from $\mu^\xi_{\Delta,p,q}$.
		\end{enumerate}
		Let $\pi$ be the distribution of
		\[(\{A_{\theta_1}\setminus E(\Delta)\} \cup A_\gamma,\\ 
		\{A_{\theta_2}\setminus E(\Delta)\}\cup A_\gamma,\{A_w\setminus E(\Delta)\}\cup A_\gamma)\]
		after step (iii), where $E(\Delta)$ denotes the set of edges with both endpoints in $\Delta$.
		
		A straightforward calculation reveals that $A_{\theta_2}$ has law $\mu^{\theta_2}$, and thus after step (ii) the distribution of $(A_{\theta_1},A_{\theta_2},A_w)$ has all the right marginals. Moreover, since $\mu_1$ and $\mu_2$ are monotone couplings, $A_{\theta_1} \subseteq A_w$ and $A_{\theta_2} \subseteq A_w$.
		
		We argue next that replacing the configuration in $\Delta$ with $A_\gamma$ in step (iii) has no effect on the distribution. For this, let $\xi_1$ (resp., $\xi_2$) be the boundary condition induced in $\Delta$ by $\eta$ and the configuration of $A_{\theta_1}$ (resp., $A_{\theta_2}$) in $\Gamma^c(A_w,B)$. (Note that $\xi$, $\xi_1$, $\xi_2$ are partitions of $\boundary \Delta \cup (\Delta \cap \boundary \LL)$.)
		
		We show that $\xi=\xi_1=\xi_2$. This is easy to see when $\Delta \cap \boundary \LL = \emptyset$, since in this case all three of them are the free boundary condition on $\Delta$. This is because $A_{\theta_1} \subseteq A_w$, $A_{\theta_2} \subseteq A_w$ and every edge from $\boundary \Delta$ to $\Delta^c$ is closed in $A_w$. 
		
		When $\Delta \cap \boundary \LL$ is not trivial, only vertices in $\Delta \cap \boundary \LL$ may be wired in $\xi$. Property (P1), together with the fact that every edge from $\boundary \Delta$ to $\Delta^c$ is closed in $A_w$, implies that two vertices from $\Delta \cap \boundary \LL$ are wired in $\xi$ iff they are wired in $\eta$. The same holds for $\xi_1$ and $\xi_2$; hence, $\xi=\xi_1=\xi_2$. (Note that this argument is essentially the same as the one used in the proof of Lemma \ref{lemma:propagation:dt} to show that the boundary conditions induced in $R_t$ by $X_t$ and $Y_t$ are the same.) 
		
		Finally, the {\it domain Markov property} of random-cluster measures (see, e.g., \cite{Grimmett}) ensures that indeed replacing $\Delta$ with $A_\gamma$ has no effect on the distribution. Hence, $\pi$ is a coupling of the measures $\mu^{\theta_1}$, $\mu^{\theta_2}$, and $\mu^w$ with all the desired properties.\end{proof}
	
	\section{Mixing time upper bound in the sub-critical regime}
	\label{section:mixing}
	
	In this section we prove our main result:~the upper bound for the mixing time in Theorem \ref{thm:mainintro} for $p < p_c(q)$. 
	We state two theorems whose combination establishes the desired upper bound for $p < p_c(q)$. In Theorem \ref{thm:crude-bound} we show that spatial mixing for the class of side-homogeneous boundary conditions, as established in Section \ref{section:boundary-conditions}, implies a bound of $O(n^2\log n(\log\log n)^2)$ for the mixing time of the Glauber dynamics on $\LL = (\LL_n,E_n)$, for any $n$ and any side-homogeneous boundary condition. The proof is inductive and makes crucial use of property (P2) of side-homogeneous boundary conditions, which ensures that for any $e \in E_n$ and $r \ge 1$, the boundary condition induced in $B(e,r)$ by the events $\{E^c=1\}$ or $\{E^c=0\}$ is also side-homogeneous.
	
	In Theorem \ref{thm:boosting} we show that a sufficiently good upper bound on the mixing time of the Glauber dynamics---in fact $O\left(n^{2.25}/\log n\right)$ suffices---can be bootstrapped to the desired upper bound of $O(n^2\log n)$. The proof of Theorem \ref{thm:boosting} crucially uses the bounds on the speed of propagation of disagreements from Section \ref{section:propagation}. Our proofs are inspired by those in the spin systems literature, in particular those in \cite{DSVW,Mart,MS}.
	
	\begin{thm}\label{thm:crude-bound} Let $p < p_c(q)$, $q \ge 1$ and let $\eta$ be a side-homogeneous boundary condition for $\LL = (\LL_n,E_n)$. There exists a fixed constant $C>0$ such that for all $n$, the mixing time of the Glauber dynamics in $\LL$ is at most $T(m) = C m \log m (\log \log m)^2$, where $m := |E_n| = \Theta(n^2)$.
	\end{thm}
	
	\begin{proof}
		We bound the coupling time $\Tcoup$ of the identity coupling; the result then follows from the fact that $\taumix \le \Tcoup$. Consider two copies $\{X_t\}$, $\{Y_t\}$  of the Glauber dynamics coupled with the identity coupling. We may assume $X_0 = E_n$ and $Y_0 = \emptyset$, since we know from Section \ref{section:prelim} that this is the worst pair of starting configurations. We prove that 
		\[\Pr[\,X_T(e) \neq Y_T(e)\,] = O\left(m^{-2}\right)\]
		for $T=T(m)$ and for all $e \in E_n$. The bound for the coupling time then follows from a union bound over the $m$ edges.
		
		To bound $\Pr[\,X_T(e) \neq Y_T(e)\,]$, we introduce two additional copies $\{Z^+_t\}$, $\{Z^-_t\}$ of the Glauber dynamics. These two copies will only update edges with both endpoints in the box $B = B(e,r)$, for a suitable $r$ we choose later. We set $Z^+_0 = X_0 = E_n$ and $Z^-_0 = Y_0 = \emptyset$. The four Markov chains $\{X_t\}$, $\{Y_t\}$, $\{Z^+_t\}$ and $\{Z^-_t\}$ are coupled with the identity coupling, and the updates outside $B$ are ignored by $\{Z^+_t\}$ and $\{Z^-_t\}$. By monotonicity of the identity coupling, we have
		$Z^-_t \subseteq Y_t \subseteq X_t\subseteq Z^+_t$
		for all $t \ge 0$. Hence,
		\begin{flalign}
		\Pr[\,X_t(e) \neq Y_t(e)\,] &\le \Pr[\,Z^+_t(e) \neq Z^-_t(e) \,] \nonumber\\
		&= \Pr[\,Z^+_t(e)=1\,] - \Pr[\,Z^-_t(e)=1\,]. \nonumber
		\end{flalign}
		
		The stationary distributions of $\{Z^+_t\}$ and $\{Z^-_t\}$ are $\mu^\eta_{\LL}(\,\cdot\,|\, E^c = 1 \,)$ and $\mu^\eta_{\LL}(\,\cdot\,|\, E^c = 0)\,$, respectively, where as usual $\{E^c=1\}$ (resp., $\{E^c=0\}$) denotes the event that all the edges in $E_n \!\setminus\! E(e,r)$ are open (resp., closed).
		The triangle inequality then implies
		\begin{align}
		\Pr [\, X_t(e) \neq Y_t(e)\,] 
		~\le~ & \left|\Pr[\,Z^+_t(e)=1\,] - \mu^\eta_{\LL}(\,e=1\,|\, E^c = 1\,)\right|\label{mixing:plus-bound}\\
		& +\left|\mu^\eta_{\LL}(\,e=1\,|\, E^c = 1\,) - \mu^\eta_{\LL}(\,e=1\,|\, E^c = 0\,)\right| \label{mixing:smp-bound}\\
		& +\left|\mu^\eta_{\LL}(\,e=1\,|\, E^c = 0\,) - \Pr[\,Z^-_t(e)=1\,] \right|. \label{mixing:minus-bound}
		\end{align}
		
		The chains $\{Z^+_t\}$ and $\{Z^-_t\}$ are (lazy) Glauber dynamics on the smaller square box $B$. Moreover, the boundary conditions induced in $B$ by $\eta$ and the events $\{E^c=1\}$, $\{E^c=0\}$ are side-homogeneous. Hence, we proceed inductively. 
		
		First note that for any fixed $m_0$, the result holds for all square boxes of volume at most $m_0$ by simply adjusting the constant $C = C(m_0)$. Now, let $r = c \log m$ for some constant $c$ we choose later, and assume the mixing time bound holds for square boxes of volume $M := |E(e,r)|$ with side-homogeneous boundary condition. After $T(m)$ steps, the expected number of updates in $B$ is 
		\[T(m) \frac{M}{m} =  C M \log m (\log \log m)^2  \ge \lceil 4 \log_2 m \rceil T(M),\]
		where we choose $m_0$ such that the last inequality holds for all $m > m_0$. Hence, a Chernoff bound implies that the number of updates in $B$ is at least $\lceil 2 \log_2 m \rceil T(M)$ with probability at least $1-m^{-2}$.
		
		The induction hypothesis implies that the mixing time of $\{Z^+_t\}$ is at most $T(M)$. Hence, if indeed $\lceil 2 \log_2 m \rceil T(M)$ updates happen in $B$, then 
		\[||\,Z^+_t(\cdot) - \mu^\eta_{\LL}(\,\cdot\,|\, E^c = 1\,)\,||_{\textsc{TV}} \le \frac{1}{m^2}.\]
		(Here we used the fact from Section \ref{section:prelim} that $\taumix(\varepsilon) \le \lceil \log_2 \varepsilon^{-1} \rceil\,\taumix$.) Combining this with the above Chernoff bound, we have
		\begin{equation}\label{mixing:localization-bounds}
		\left|\Pr[\,Z^+_t(e)=1\,] - \mu^\eta_{\LL}(\,e=1\,|\, E^c = 1)\right| \le \frac{2}{m^2}. \nonumber
		\end{equation}
		The quantity in (\ref{mixing:minus-bound}) is bounded similarly. 
		
		Finally, taking $c = 2/\Ll$, the spatial mixing property (Lemma \ref{lemma:bc-smp}) implies that (\ref{mixing:smp-bound}) is at most $1/m^2$. Putting these bounds together we get
		\begin{equation}
		\Pr[\,X_t(e) \neq Y_t(e)\,] \le \frac{5}{m^2},
		\nonumber
		\end{equation}
		as desired.\end{proof}
	
	\begin{thm}\label{thm:boosting}
		Let $p < p_c(q)$, $q \ge 1$ and let $m_0$, $c$ be sufficiently large and sufficiently small positive constants, respectively. Assume that the mixing time of the Glauber dynamics in any square box of volume $m_0$ with side-homogeneous boundary conditions is at most $\frac{c\, m_0^{9/8}}{\log_2 m_0}$.
		Then, the mixing time of the Glauber dynamics in $\LL$ with side-homogeneous boundary conditions is $O(n^2\log n)$.
	\end{thm}
	
	\begin{proof} Let $m \!:=\! |E_n| \!=\! \Theta(n^2)$ and the let $\eta$ be a side-homogeneous boundary condition for $\LL$. Also, let $\{X_t\}$, $\{Y_t\}$ be two copies of the Glauber dynamics in $\LL$ coupled with the identity coupling. We prove that for $1 \le k = o(m^{1/8})$, we have 
		\begin{equation}
		\Pr[X_{k m}(e) \neq Y_{k m}(e)] \le {\e}^{-\Omega(k)} \label{eq:mixing-up-boosting}
		\end{equation}
		for any $e \in E_n$ and any pair of initial configurations $X_0,Y_0$. Hence, for some $k=O(\log m)$ we have $\Pr[X_{k m}(e) \neq Y_{k m}(e)] \le 1/(4m)$, and a union bound over the edges implies $\taumix \le \Tcoup = O(m\log m)$, as desired.
		
		By the monotonicity of the identity coupling discussed in Section \ref{section:prelim},
		it suffices to bound $\Pr[X_{k m}(e) \neq Y_{k m}(e)]$ for the case when $X_0=E_n$ and $Y_0=\emptyset$.
		For this, let $\{Z_t\}$ by another instance of the Glauber dynamics in $\LL$, with $Z_0$ sampled according to $\mu^\eta$, and coupled with $\{X_t\}$ and $\{Y_t\}$
		via the identity coupling. Then:
		\begin{equation}
		\label{eq:triangle}
		\Pr[X_{k m}(e) \neq Y_{k m}(e)] \le \Pr[X_{k m}(e) \neq Z_{k m}(e)] +\Pr[Z_{k m}(e) \neq Y_{k m}(e)]. 
		\end{equation}
		Let
		\[
		\rho(k) := \max_{e \in E_n}\, \max\{ \Pr[\,X_{km}(e) \neq Z_{km}(e)\,],\Pr[Z_{k m}(e) \neq Y_{k m}(e)]\},
		\]
		so that $\Pr[X_{k m}(e) \neq Y_{k m}(e)] \le 2\rho(k)$.		
		Let ${\cal F}_k$ be the event $\{X_{km}(B(e,r))=Z_{km}(B(e,r))\}$ 
		for a fixed $r$ we choose later,
		and, as in~\eqref {eq:only-small-cmpt}, for $\ell = r^{1/4}-1$ let 
		$$
		{\cal E}_{\ell,t} := \{u \tightoverset{\scriptscriptstyle{Z_t}}{\nleftrightarrow} v~~~\forall u,v \in B(e,r)~\text{s.t.}~d(u,v)>\ell \}.
		$$
		Then,
		\begin{align}
		\Pr[X_{2k m}(e) \neq Z_{2k m}(e)]  \le &
		\Pr[X_{2k m}(e) \neq Z_{2k m}(e) \,|\, \neg{\cal F}_k]\Pr[\neg{\cal F}_k] \label{eq:triang:1}\\
		&+ 	\Pr[X_{2k m}(e) \neq Z_{2k m}(e), {\cal F}_k, \bigcap\nolimits_{t=km}^{2km} {\cal E}_{\ell,t}] \\
		&+ 
		\Pr[\bigcup\nolimits_{t=km}^{2km} \neg{\cal E}_{\ell,t}].
		\end{align}
		The monotonicity of the identity coupling implies that 
		$$
		\Pr[\, X_{2km}(e) \neq Z_{2km}(e)\,|\,\neg{\cal F}_k \,] \le 2\rho(k),
		$$
		and a union bound over the edges in $B(e,r)$ yields that $\Pr[\,\neg{\cal F}_k \,] \le \Theta(r^2) \rho(k)$. Hence, the right-hand side of~\eqref{eq:triang:1} is at most $\Theta(r^2) \rho(k)^2$.
		From~\eqref{propagation:first-term}, we also have that if $k \le r^{1/4}/(4e^2)$, then
		$$
		\Pr[X_{2k m}(e) \neq Z_{2k m}(e), {\cal F}_k, \bigcap\nolimits_{t=km}^{2km} {\cal E}_{\ell,t}] \le \frac{\e}{\e-1} \cdot {\e}^{-r^{3/4}},		
		$$
		and from~\eqref{eq:neg:small-cmpt} that
		$$
		\Pr[\bigcup\nolimits_{t=km}^{2km} \neg{\cal E}_{\ell,t}] \le O(r^{6.25}) \cdot {\e}^{-\Omega(r^{1/4})} + {\e}^{-\Omega(r^{2})}.
		$$		
		Putting these bounds together and setting $r = \Theta(k^{4})$, we obtain that for any $e \in E_n$
		$$
		\Pr[X_{2k m}(e) \neq Z_{2k m}(e)] \le C k^8 \rho^2(k) + {\e}^{-\Ll k},
		$$
		for suitable constants $C>1$ and $\Ll > 0$. 
		The same bound can be analogously derived for $\Pr[Z_{2k m}(e) \neq Y_{2k m}(e)]$, so that
		\begin{equation}
		\rho(2k) \le C k^8 \rho^2(k) + {\e}^{-\Ll k}.\label{eq:mixing-time-ub-rec}\nonumber
		\end{equation}		
		(Note that since $r = \Theta(k^{4})$ and $r < n$, our proof of inequality (\ref{eq:mixing-up-boosting}) does not hold for 
		arbitrarily large $k$; hence the restriction $k = o(m^{1/8})$.)
		
		Now, let
		\[\phi(k) := 2^8(C k^8+1)\max\{\rho(k),{\e}^{-\Ll k/2}\}.\]
		We show next that $\phi(2k) \le \phi(k)^2$. For this observe that $2^8(C(2k)^8+1)\rho(2k) \le \phi(k)^2$. Hence, if $\rho(2k) \ge e^{-\Ll k}$, we get $\phi(2k) \le \phi(k)^2$ directly. Otherwise, we have 
		\[\phi(2k) \le (2^8)^2(C k^8+1){\e}^{-\Ll k} \le \phi(k)^2.\] 
		Thus, for any integer $\alpha > 0$, we get $\phi(k) \le \phi(k/2^\alpha)^{2^\alpha}$. The result follows from the following fact, which provides a stopping point for this recurrence. 
		
		\begin{Claim}\label{claim:base-case}
			Let $l = \frac{m_0^{1/8}}{2(8C{\e})^{1/8}}$. Then $\rho(l) \le \frac{1}{2^8\e (C l^8 + 1)}$ for a sufficiently large~$m_0$.
		\end{Claim}
		\noindent
		As a result, $\phi(l) \le 1/{\e}$ for a sufficiently large constant $l$, and thus $\rho(k) \le \phi(k) \le \exp(-k/l)$.
		Since
		$\Pr[X_{k m}(e) \neq Y_{k m}(e)] \le 2 \rho(k)$,
		the result follows.
	\end{proof}
	
	\noindent
	We conclude this section with the proof of Claim \ref{claim:base-case}. The proof is similar to that of Theorem 5.1 and makes crucial use of the hypothesis on the mixing time in square boxes of volume $m_0$.
	
	\begin{proof}	Let $e \in E_n$ and choose $r^\prime$ such that $|E(e,r^\prime)|=m_0$. (Note that as a result $r^\prime=\Theta(m_0^{1/2})$.) The proof proceeds along the same lines as that of Theorem \ref{thm:crude-bound}. In fact we consider the same processes $\{Z^+_t\}$, $\{Z^-_t\}$, where $Z^+_0 = X_0 = E_n$, $Z_0^- = Y_0=\emptyset$ and $\{Z^+_t\}$, $\{Z^-_t\}$ only update edges with both endpoints in $B(e,r^\prime)$. As before, we couple the four chains $\{X_t\}$, $\{Y_t\}$, $\{Z^+_t\}$, $\{Z^-_t\}$ with the identity coupling, ignoring the updates outside $B(e,r^\prime)$ in $\{Z^+_t\}$ and $\{Z^-_t\}$. The monotonicity of the identity coupling then implies that $Z^-_t \subseteq Y_t \subseteq X_t \subseteq Z^+_t$ for all $t\ge 0$. Hence, we obtain inequality (\ref{mixing:plus-bound})-(\ref{mixing:minus-bound}). 
		
		Lemma \ref{lemma:bc-smp} implies that (\ref{mixing:smp-bound}) is at most $\exp(-\Omega(r^\prime))$. To bound (\ref{mixing:plus-bound}), note that if we run the identity coupling for $l m$ steps, a Chernoff bound implies that with probability at least $1-\exp(-l m_0/3)$ the number of updates in $B(e,r^\prime)$ is at least $A^{-1} m_0^{9/8}$, where $A = (8C{\e})^{1/8}$. If indeed this many steps hit $B(e,r^\prime)$, then the hypothesis of the theorem  (with $c=A^{-1}$) implies
		\begin{align}
		\left|\Pr[\,Z^+_t(e)=1\,] -\mu^\eta_{\LL}(\,e=1\,|\, E^c = 1)\right|
		\le ||\,Z^+_t(\cdot) - \mu^\eta_{\LL}(\,\cdot\,|\, E^c = 1\,)\,||_{\textsc{TV}} \le \frac{1}{m_0}.\nonumber
		\end{align}
		(Here we also used the fact that $\taumix(\varepsilon) \le \lceil \log_2 \varepsilon^{-1} \rceil \taumix$.) We do the same to bound (\ref{mixing:minus-bound}), and then
		\begin{align}
		\Pr[\, X_{lm}(e) \neq Y_{lm}(e) \,] &\le \frac{2}{m_0} + {e}^{-\Omega(m_0^{9/8})} + {e}^{-\Omega(m_0^{1/2})} \notag\\
		&\le \frac{4}{m_0} = \frac{4}{(2 A l)^{8}} \le \frac{1}{2^8\e (C l^8 + 1)},\nonumber
		\end{align}
		for a sufficiently large constant $m_0$. Hence, $\phi(l) \le 1/\e$ for a sufficiently large $l$, as desired.\end{proof}

	\section{Mixing time lower bound in the sub-critical regime}
	\label{section:mixing-lower}
	
	In this section we prove the lower bound from Theorem \ref{thm:mainintro} for $p < p_c(q)$.
	(The lower bound for $p > p_c(q)$ is derived in Section \ref{section:super-critical}.)
	In the setting of spin systems, \cite{HS} provides a general mixing time lower bound for Glauber dynamics. As mentioned earlier, the random-cluster model is not a spin system in the usual sense because of the long range interactions, but we are still able to adapt the techniques in \cite{HS} to the random-cluster setting. In fact, our proof follows closely the argument in the proof of Theorem 4.1 in \cite{HS}, the main difference being that we require a more subtle choice of the starting configuration to limit the effect of the long range interactions.  We derive the following theorem.
	
	\begin{thm}\label{thm:lower-bound} Let $p < p_c(q)$, $q \ge 1$ and let $\eta$ be a side-homogeneous boundary condition for $\LL = (\LL_n,E_n)$. The mixing time of the Glauber dynamics in $\LL$ is $\Omega(n^2 \log n)$.
	\end{thm}
	
	\noindent
	It is convenient to carry out our proof in continuous time.
	The continuous time Glauber dynamics is obtained by adding a rate 1 Poisson clock to each edge; when the clock at edge $e$ rings, $e$ is updated as in discrete time.
	
	The switch to continuous time requires us to extend the bound in Section \ref{section:propagation} for the speed of propagation of disagreements to the continuous time dynamics.
	In addition, we will require slightly different assumptions about the initial configuration $Y_0$.
	This is established in the following lemma, whose proof is very similar to that of Lemma \ref{lemma:propagation:dt} (only requiring minor adjustments) and is deferred to Appendix 
	A.
	
	\begin{lem} \label{lemma:propagation:ct}
		Let $p < p_c(q)$, $q \ge 1$ and let $\eta$ be a side-homogeneous boundary condition for $\LL=(\LL_n,E_n)$. Also, let $B=B(e,r)$ for some $e\in E_n$ and $r \ge 1$. Consider two copies $\{X_t\}$, $\{Y_t\}$ of the continuous time Glauber dynamics on $\LL$ such that:
		\begin{itemize}
			\setlength\itemsep{1em}
			\item[$\bullet$] $X_0(B)=Y_0(B)$;
			\item[$\bullet$] $Y_0(B)$ has law $\mu^0_{B,p,q}(\,\cdot\,|\,e=b)$ for some $b \in \{0,1\}$;
			\item[$\bullet$] $Y_0(e^\prime)=0$ for all $e^\prime \in E^c(e,r)$ incident to $\boundary B$; and
			\item[$\bullet$] $\{Y_t\}$ only performs edge updates in $B$.
		\end{itemize}
		If the evolutions of $\{X_t\}$ and $\{Y_t\}$ are coupled using the identity coupling, then there exist absolute positive constants $c$, $C$, and $\lambda$ (independent of $r$ and $n$) such that, for all $r \ge c$ and $1 \le T \le r^{1/4}/(4{\e}^2)$, we have
		\begin{equation}
		\Pr[\,X_{T}(e) \neq Y_{T}(e)\,] ~\le~ C {\e}^{-\lambda r^{1/4}}.\nonumber
		\end{equation}	
	\end{lem}
	
	\noindent
	We are now ready to prove Theorem \ref{thm:lower-bound}.
	
	\begin{proof}	Let $\{X_t\}$ and $\{X_t^{\cal D}\}$ be copies of the continuous and discrete time Glauber dynamics in $\LL$, respectively, such that $X_0^{\cal D} = X_0$. The following standard inequality  holds for all $t\ge 0$: 
		\[||X_{t^\prime}^{\cal D}-\mu^\eta||_{\textsc{\tiny TV}} \ge ||X_t-\mu^\eta||_{\textsc{\tiny TV}}-2e^{-t^\prime},\]
		where $t^\prime = |E_n| t / 2$ (see, e.g., Proposition 2.1 in  \cite{HS}).

		We will show that $||X_T-\mu^\eta||_{\textsc{\tiny TV}} > 1/3$ for some $T = \Omega(\log n)$; as a result $||X_{T^\prime}^{\cal D}-\mu^\eta||_{\textsc{\tiny TV}} > 1/4$ for some $T' = \Omega(n^2 \log n)$ and sufficiently large $n$. This implies that the mixing time of the discrete time dynamics is $\Omega(n^2 \log n)$ as desired. 
		
		First we introduce some notation. Assume w.l.o.g.\ that $4r+1$ divides $n-1$ for some $r$ to be chosen later, and split $\LL_n$ into $(n-1)^2/(4r+1)^2$ square boxes of side length $4r+1$. Each of these boxes corresponds to $B(e,2r)$ for some edge $e \in E_n$; let $\C \subseteq E_n$ be the set of these edges. Also, let $\hat{E} := E_n \setminus \bigcup_{e \in \C} \;E(e,r)$ and let ${\cal E}$ be the event that every edge $e^\prime \in \hat{E}$ incident to $B(e,r)$ for some $e \in \C$ is closed.
		
		Let $A$, $A_{\cal E}$ be random-cluster configurations sampled from  $\mu^\eta$ and $\mu^\eta(\,\cdot\,|\, {\cal E})$, respectively, and let $\beta := \E[f(A_{\cal E})]$, where $f(A_{\cal E})$ is the fraction of edges $e \in \C$ such that $A_{\cal E}(e) = b$ for some fixed $b \in \{0,1\}$. Consider the following threshold for a value of $\varepsilon > 0$ that will be chosen later:
		\begin{equation}
		\bh = \begin{cases}
		\beta+\varepsilon & \text{if $\beta < 1/2$};\\
		\beta & \text{if $\beta=1/2$};\\
		\beta-\varepsilon & \text{if $\beta > 1/2$}.
		\end{cases} \nonumber
		\end{equation}
		We pick $b$ such that $\Pr[f(A) > \bh] \le 1/2$. 
		
		As in \cite{HS}, our goal is to choose $X_0$ such that $\Pr[f(X_T) \ge \bh] \rightarrow 1$ as $n \rightarrow \infty$ for some $T=\Omega(\log n)$; then $||X_T-\mu^\eta||_{\textsc{\tiny TV}} > 1/3$ for large enough $n$, as desired. 
		
		Let $\Phi$ be the set of random-cluster configurations in $\LL$ such that $\hat{A} \in \Phi$ iff for all $e \in \C$, $\hat{A}(e)=b$ and $\hat{A}(e^\prime)=0$ for all $e'\in{\hat E}$ incident to $B(e,r)$. For each $\hat{A} \in \Phi$, let
		\begin{equation}\label{eq:lb-initial}
		\pi_0(\hat{A}) := \frac{\mu^\eta(\hat{A})}{\mu^\eta(\Phi)}.
		\end{equation}
		The starting configuration $X_0$ is sampled from $\pi_0$.
		
		Consider now an auxiliary copy $\{Y_t\}$ of the continuous time Glauber dynamics such that $Y_0 = X_0$. The two chains $\{X_t\}$,$\{Y_t\}$ are coupled using the identity coupling, except that $\{Y_t\}$ does not update edges in $\hat{E}$. First we establish a bound for $\Pr[f(Y_T) \le \bh + \varepsilon/2]$. To do this, we use the following monotonicity property which is a straightforward consequence of Lemma 3.5 in \cite{HS}.
		
		\begin{fact}\label{fact:monotonicity}
			For each $e \in \C$, let $\alpha_e := \mu^{0}_{B,p,q}(e=b)$ where $B=B(e,r)$. Then, for all $t \ge 0$, \[\Pr[Y_t(e) = b] \ge \alpha_e + (1-\alpha_e){\e}^{-t/(1-\alpha_e)}.\]
		\end{fact}
		
		\noindent
		From this fact, we follow the steps in the proof of Theorem 4.1 in \cite{HS} to derive the following bound:
		\begin{equation}\label{lower-bound:exp}
		\E[f(Y_T)] \ge \beta + (1-\beta){\e}^{-T/(1-\beta)},
		\end{equation}
		for all $T \ge 0$. Taking $\varepsilon = 1/(4\exp(2T))$, the right hand side of (\ref{lower-bound:exp}) is at least $\bh+\varepsilon$. Also, since $Y_0$ is sampled from $\pi_0$, the configurations of the edges in $\C$ are independent in $Y_T$. Hence, $|\C|f(Y_T)$ is the sum of $|\C|$ independent Bernoulli random variables; a Chernoff bound then implies
		\begin{equation}\label{lower-bound:prob1}
		\Pr[f(Y_T) \le \bh + \varepsilon/2] \le {\e}^{-\Omega(\varepsilon^2|\C|)}.
		\end{equation}
		
		The second step in the proof is to bound $\Pr[f(X_T) \le f(Y_T)-\varepsilon/2]$. For this, we use Lemma \ref{lemma:propagation:ct}, which is tailored precisely to our setting. Thus, for all $e \in \C$ and $1 \le T \le r^{1/4}/(4{\e}^2)$, we have
		\[\Pr[X_T(e) \neq Y_T(e)] \le C {\e}^{-\lambda r^{1/4}},\]
		provided $r \ge c$, where $c$ is a sufficiently large constant. Therefore, the expected number of disagreements between $X_T$ and $Y_T$ in the set $\C$ is at most $|\C|C \exp(-\lambda r^{1/4})$, and by Markov's inequality,
		\begin{equation}\label{lower-bound:prob2}
		\Pr[f(X_T) \le f(Y_T) - \varepsilon/2] \le \frac{2C}{\varepsilon} {\e}^{-\lambda r^{1/4}}.
		\end{equation}
		Putting together the bounds in (\ref{lower-bound:prob1}) and (\ref{lower-bound:prob2}), we get
		\[\Pr[f(X_T) \le \bh] \le  \frac{2C}{\varepsilon} {\e}^{-\lambda r^{1/4}}+ {\e}^{-\Omega(\varepsilon^2|\C|)}.\]
		
		\noindent
		Finally, observe that $|\C| = \Theta(\frac{n^2}{r^2})$; thus, when $r = (\frac{1}{4} \log n)^4$ and $T = \min\{\lambda/4,1\} r^{1/4}$, we get $\Pr[f(X_T) > \bh] \rightarrow 1$ as $n \rightarrow \infty$ as desired.
	\end{proof}
	
	\section{Mixing time in the super-critical regime}
	\label{section:super-critical}
	
	In this section we prove Theorem \ref{thm:mainintro} from the Introduction for $p > p_c(q)$. We will make use of planar duality, discussed in Section \ref{section:prelim}, in order to reduce the proof to the sub-critical case.
	
	\begin{thm}\label{thm:super-critical} For $p > p_c(q)$ and $q \ge 1$, the mixing time of the Glauber dynamics on $\LL = (\LL_n,E_n)$ with free or wired boundary conditions is $\Theta(n^2 \log n)$.
	\end{thm}
	
	\begin{proof} We focus on the free boundary condition case; the wired case follows from an analogous argument. The planar dual $\LL^* = (\LL_n^*,E_n^*)$ of the graph $\LL$ consists of an $(n-1) \times (n-1)$ box with an additional outer vertex corresponding to the infinite face of $\LL$. The dual measure $\mu_{\LL^*,p^*,q}$, with
		\[p^*= \frac{q(1-p)}{p+q(1-p)},\]
		is equivalent to the measure $\mu^1_{\LL',p^*,q}$ where $\LL' = (\LL_{n+1},E_{n+1})$ (see, e.g., \cite{BDC}). Note that $p > p_c(q)$ iff $p^* < p_c(q)$. 
		
		We say that two random-cluster configurations $A$ on $\LL$ and $A' $ on $\LL'$ are {\it compatible} if the configuration resulting from $A'$ by contracting all the vertices in the boundary of $\LL'$ into a single vertex is $A^*$, the dual configuration of $A$. Note that each $A'$ has a unique compatible $A$, while each $A$ has multiple compatible $A'$ that differ only in the disposition of the edges in the boundary $\boundary \LL_{n+1}$. 
		Observe also that any edge $e'$ of $\LL'$ with at most one endpoint incident to $\boundary \LL_{n+1}$ corresponds to a unique dual edge $e^* \in E_n^*$ and thus to a unique edge $e \in E_n$.
		
		
		In order to analyze the Glauber dynamics on $\LL$ when $p > p_c(q)$, we consider instead the Glauber dynamics on $\LL'$ with parameter $p^* < p_c(q)$, which we denote $\{X_t'\}$. We shall show that $\{X_t'\}$ induces a Markov chain $\{X_t\}$ on $\LL$ which is essentially the same as the Glauber dynamics on $\LL$ with parameter $p$, and that
		the mixing times of $\{X_t\}$ and $\{X_t'\}$ are equal up to constant factors. 
		Since $p^* < p_c(q)$, the results in Sections \ref{section:mixing} and  \ref{section:mixing-lower} imply that mixing time $\{X_t'\}$ (and hence of $\{X_t\}$) is $\Theta(n^2\log n)$. 
		
		To define the induced dynamics, let $e_t'$ be the edge chosen u.a.r. from $E_{n+1}$ at time $t$ by $\{X_t'\}$, and let $e_t$ be the corresponding edge in $\LL$ if there is one. $X_{t+1}$ is obtained from $X_t$ as follows:
		\begin{enumerate}[(i)]
			\setlength\itemsep{1em}
			\item if both endpoints of $e_t'$ are in $\boundary \LL_{n+1}$, then $X_{t+1} = X_t$;
			\item else if $X_{t+1}' = X_{t+1}' \cup \{e_t'\}$, then $X_{t+1}= X_t \setminus \{e_t\}$;
			\item else if $X_{t+1}' = X_{t+1}' \setminus \{e_t'\}$, then $X_{t+1}= X_t \cup \{e_t\}$.
		\end{enumerate}
		The initial configuration $X_0$ is the unique configuration compatible with $X_0'$. 
		
		We show first that $\{X_t\}$ is in fact a lazy version of the Glauber dynamics on $\LL$. To see this, note that 
		$X_{t+1} = X_{t}$ whenever both endpoints of $e_t'$ are in $\boundary \LL_{n+1}$. Otherwise, it is straightforward to check that		
		$e_t \in X_t$ is a cut edge iff $e_t' \in X_t'$ is not a cut edge. Hence, $X_{t+1} = X_{t} \cup \{e_t\}$ iff $X_{t+1}' = X_{t} \setminus \{e_t'\}$ and thus $X_{t+1} = X_{t} \cup \{e_t\}$ with probability:
		\[\left\{\begin{array}{ll}
		1 - \frac{p^*}{q(1-p^*)+p^*} = p & ~~~~\mbox{if $e_t'$ is a cut edge;} \\
		\\
		1 - p^* =  \frac{p}{q(1-p)+p} & ~~~~\mbox{otherwise.}
		\end{array}\right.\]
		This implies that $\{X_t\}$ does not move with probability $\Theta(n^{-1})$, and otherwise evolves exactly like the Glauber dynamics on $\LL$. Hence, it is sufficient for us to establish the mixing time of $\{X_t\}$.
		To do this, we show that the mixing times of $\{X_t\}$ and $\{X_t'\}$ are essentially the same.
		
		Let $\Omega_{\textsc{RC}}$ be the set of random-cluster configurations on $\LL$, and let ${\cal C}(A)$ be the set of configurations compatible with a configuration $A$ on $\LL$. The first observation is that when $\{X_t'\}$ mixes, so does $\{X_t\}$. This follows from:
		\begin{flalign}
		||\,X_t(\cdot) - \mu_{\LL,p,q}(\cdot)\,||_{\textsc{TV}} &= \frac{1}{2} \sum_{A \in \Omega_{\textsc{RC}}} |X_t(A)-\mu_{\LL,p,q}(A)| \nonumber\\
		&=\frac{1}{2} \sum_{A \in \Omega_{\textsc{RC}}} |X_t(A)-\mu_{\LL^*,p^*,q}(A^*)| \nonumber\\
		&\le \frac{1}{2}\sum_{A \in \Omega_{\textsc{RC}}} \sum_{A' \in {\cal C}(A)} \left|X_t'(A')-\mu^1_{\LL',p^*,q}(A')\right|
		\nonumber\\
		&= ||\,X_t'(\cdot) - \mu^1_{\LL',p^*,q}(\cdot)\,||_{\textsc{TV}}, \nonumber
		\end{flalign}
		where in the first and last equality we use the definition of total variation distance, in the second equality we use planar duality, and the third inequality follows from the triangle inequality and the correspondence between the configurations of $\LL$ and $\LL'$. Hence, by the results in Section \ref{section:mixing}, the mixing time of $\{X_t\}$ is $O(n^2\log n)$. 
		
		We show next that the mixing time of $\{X_t\}$ is $\Omega(n^2 \log n)$. For this, note that in Theorem \ref{thm:lower-bound} we showed that there is an initial distribution $\pi_0$ for $X_0'$, defined in (\ref{eq:lb-initial}), such that 
		\[||X_T'(\cdot)-\mu^1_{\LL',p^*,q}(\cdot)||_{\textsc{TV}} > 1/4\]
		for some $T=\Omega(n^2\log n)$. We will prove that when $X_0'$ is sampled from $\pi_0$, then 
		\begin{equation}\label{equation:lb-observation2-1}
		||X_t'(\cdot)-\mu^1_{\LL',p^*,q}(\cdot)||_{\textsc{TV}} = ||X_t(\cdot)-\mu_{\LL,p,q}(\cdot)||_{\textsc{TV}} \nonumber
		\end{equation}
		for all $t \ge 0$. To show this we introduce some additional notation. 
		
		Let $\LL'' := (\LL_{n+1},E_{n+1}\setminus \boundary E_{n+1})$, where $\boundary E_{n+1} \subseteq E_{n+1}$ is the set of edges with both endpoints in $\boundary \LL_{n+1}$. Also, for any random-cluster configuration $A'$ on $\LL'$, we use $\boundary A'$ to denote the random-cluster configuration induced in $\boundary\LL'$ by $A'$.
		
		Under the wired boundary condition, we have that for any $e \in \boundary \LL'$, $\mu^1_{\LL',p^*,q} (e = 1) = p^*$. Hence, $\mu^1_{\LL',p^*,q}$ is the product measure of the distributions $\mu^1_{\LL'',p^*,q}$ and $\mu_{\boundary \LL',p^*,1}$; the latter is the distribution on $\boundary \LL'$ where every edge is sampled independently with probability $p^*$. Thus we have
		\[\mu^1_{\LL',p^*,q}(A') =  \mu^1_{\LL'',p^*,q}(A'\setminus\boundary A')\cdot\mu_{\boundary \LL',p^*,1}(\boundary A').\]
		By the correspondence between the configurations of $\LL^*$ and $\LL'$,  we have that $\mu^1_{\LL'',p^*,q}(A'\setminus\boundary A') = \mu_{\LL^*,p^*,q}(A^*)$. (As in Section~\ref{section:prelim}, $A^*$ denotes the dual of the unique configuration $A$ compatible with $A'$.) Moreover, by planar duality $\mu_{\LL^*,p^*,q}(A^*) = \mu_{\LL,p,q}(A)$, and thus 
		\begin{equation}
		\mu^1_{\LL',p^*,q}(A') =  \mu_{\LL,p,q}(A) \cdot \mu_{\boundary \LL',p^*,1}(\boundary A').\label{eq:lower-bound-1}
		\end{equation}
		
		Also, under the wired boundary condition the configuration on the boundary of $X_0'$ is sampled according to $\mu_{\boundary \LL',p^*,1}$. Hence, the distribution on the boundary of $X_t'$ has law $\mu_{\boundary \LL',p^*,1}$ for all $t \ge 0$. Thus, 
		\begin{equation}
		X_t'(A') = \mu_{\boundary \LL',p^*,1}(\boundary A') \cdot X_t'(A'\setminus\boundary A') 
		= \mu_{\boundary \LL',p^*,1}(\boundary A') \cdot X_t(A).\label{eq:lower-bound-2}
		\end{equation}
		Hence,
		\begin{flalign}
		||\,X_t'(\cdot) - \mu^1_{\LL',p^*,q}(\cdot)\,||_{\textsc{TV}} 
		&= \frac{1}{2}\sum_{A \in \Omega_{\textsc{RC}}} \sum_{A' \in {\cal C}(A)} \left|X_t'(A')-\mu^1_{\LL',p^*,q}(A')\right|
		\nonumber\\
		&= \frac{1}{2}\sum_{A \in \Omega_{\textsc{RC}}} \sum_{A' \in {\cal C}(A)} \mu_{\boundary \LL',p^*,1}(\boundary A')\left|X_t(A)-\mu_{\LL,p,q}(A)\right|\nonumber\\
		&=  ||\,X_t(\cdot)-\mu_{\LL,p,q}(\cdot)\,||_{\textsc{TV}}, \label{equation:lb-observation2}\nonumber
		\end{flalign}
		where in the first and last equality we used the definition of total variation distance and the second follows from (\ref{eq:lower-bound-1}) and (\ref{eq:lower-bound-2}).
		
		The results in Section \ref{section:mixing-lower} then imply that the mixing time of $\{X_t\}$ is $\Omega(n^2\log n)$. Consequently, the Glauber dynamics on $\LL$ with $p > p_c(q)$ mixes in $\Theta(n^2\log n)$ steps, as desired. \end{proof}
	
	\section*{Acknowledgments}
	
	The authors would like to thank Fabio Martinelli and Allan Sly for helpful suggestions. We would like also to thank the anonymous referees whose valuable comments improved the exposition of the results.
	

		
	\section*{A Proof of Lemma \ref{lemma:propagation:ct}}\label{A.lemma:propagation:ct}
	
	We show first that the measure that results from conditioning on the state of a single edge maintains the exponential decay of finite volume connectivities (\ref{prelim:edc-finite}).
	
	\begin{fact}\label{fact:exp-decay-extension}
		Let $p < p_c(q)$, $q \ge 1$, and let $\eta$ be a boundary condition for $\LL = (\LL_n,E_n)$.  Consider a copy $\{Y_t\}$ of the continuous time Glauber dynamics on $\LL$, and assume $Y_0$ is sampled from the distribution $\mu^\eta_{\LL,p,q}(\,\cdot\,|\,e=b)$, for some $e \in E_n$ and $b \in \{0,1\}$. Then, for all $u,v \in \LL_n$, there exists positive constant $C$ and $\Ll$ such that
		\[\Pr[u \tightoversetsub{\scriptscriptstyle Y_t}{\leftrightarrow} v] \le C {\rm e}^{-\Ll d(u,v)},\]
		where $u \tightoversetsub{\scriptscriptstyle Y_t}{\leftrightarrow} v$ denotes the event that $u$ and $v$ are connected by a path of open edges in $Y_t$.
	\end{fact}
	
	\begin{proof} Let $\{Z_t\}$ be a second instance of the continuous time Glauber dynamics. The evolution of $\{Z_t\}$ is coupled with that of $\{Y_t\}$ via the identity coupling, except that $\{Z_t\}$ never updates the edge $e$. The initial configuration of $\{Z_t\}$ is sampled according to the distribution $\mu^{\eta}(\,\cdot\,|\,e=1)$ such that $Y_0 \subseteq Z_0$. This is always possible because $\mu^{\eta}(\,\cdot\,|\,e=1) \succeq \mu^{\eta}$. Then, $Y_t \subseteq Z_t$ and $Z_t$ has law $\mu^{\eta}(\,\cdot\,|\,e=1)$ for all $t \ge 0$. We establish that the measure $\mu^{\eta}(\,\cdot\,|\,e=1)$ has exponential decay of finite volume connectivities and thus so does the distribution of $Y_t$ for all $t \ge 0$. By (\ref{prelim:edc-finite}), for all $u,v \in \LL_n$, we have
		\[\mu^{\eta}(\,u \;{\leftrightarrow}\; v\,|\,e=1)\mu^\eta(e=1) \le \mu^\eta(u\;{\leftrightarrow}\; v) \le C {\rm e}^{-\Ll d(u,v)},\]
		where $C,\Ll$ are positive constants. If $p' = \frac{p}{q(1-p)+p}$, then $\mu^{\eta} \succeq \mu^{\eta}_{\LL,p',1}$ (see, e.g., \cite{Fort}), and thus $\mu^\eta(e=1) \ge p'$. Since $q \ge 1$,
		\[\mu^{\eta}(\,u\;{\leftrightarrow}\; v\,|\,e=1) \le \frac{qC}{p} {\rm e}^{-\Ll d(u,v)}.\]
		The result then follows immediately when $p = \Omega(1)$. Otherwise, the measure $\mu^{\eta}$ is stochastically dominated by any random-cluster measure $\mu^{\eta}_{\LL,p'',q}$ with $p''= \Omega(1)$, for which we just established exponential decay of finite volume connectivities; the result follows by monotonicity. \end{proof}
	
	\noindent
	We are now ready to prove the lemma.  
	
	\begin{proof} Let $Q_t$ be the random time at which the $t$-th edge is updated by the identity coupling. For some fixed $\ell$ to be chosen later, and each $t \ge 0$, consider the event
		\[{\cal E}_{\ell,t} := \{u \overset{\scriptscriptstyle{Y_{Q_t}}}{\nleftrightarrow} v~~~\forall u,v \in B~\text{s.t.}~d(u,v)>\ell \},\]
		where $u\overset{\scriptscriptstyle{Y_{Q_t}}}{\nleftrightarrow}v$ denotes the event that $u$ and $v$ are not connected by a path in $Y_{Q_t}(B)$.
		Also, let ${\cal E}_\ell := \bigcap_{t:Q_t \le T} {\cal E}_{\ell,t}$. Then,
		\begin{align}\label{propagation:main-ineq-ct}
		\Pr[\,X_{T}(e) \neq Y_{T}(e)\,] 
		\le \Pr[\,X_{T}(e) \neq Y_{T}(e) \,|\, {\cal E}_\ell\,] + \Pr[\,\neg {\cal E}_\ell\,]
		\end{align}
		(cf. equation (\ref{propagation:main-ineq})). We bound each term on the right hand side of (\ref{propagation:main-ineq-ct}) separately.
		
		Conditioned on the event ${\cal E}_\ell$, a witness for the fact that $X_{T}(e) \neq Y_{T}(e)$ can be constructed as in discrete time. However, the probability that a given witness of length $L$ is updated by the continuous time dynamics is instead bounded using the following fact from \cite{HS}.
		
		\begin{fact}\label{fact:poisson-clocks} Consider $L$ independent rate $1$ Poisson clocks. Then, the probability that there is an increasing sequence of times $0 \le t_1 < ... < t_L \le T$ such that clock $i$ rings at time $t_i$ is at most $\left(\frac{eT}{L}\right)^L$.
		\end{fact}
		
		\noindent
		Recall from Section \ref{section:propagation} that the number of witnesses of length $L$ is at most $(4(\ell+1)^2)^L$  (crudely). Hence, following the same steps as in the proof Lemma \ref{lemma:propagation:dt}, and taking $\ell = r^{1/4}-1$, we get
		\begin{equation}\label{propagation:first-term-ct}
		\Pr[\,X_{T}(e) \neq Y_{T}(e) \,|\, {\cal E}_\ell\,] \le \frac{\e}{\e-1} \cdot {\e}^{-r^{3/4}},
		\end{equation}
		using the fact that $T \le r^{1/4}/(4{\e}^2)$ (cf. equation (\ref{propagation:first-term})).
		
		To bound the second term on the right hand side of (\ref{propagation:main-ineq-ct}), let $N$ be the number of edge updates in $B$ up to time $T$. Observe that $N$ is a Poisson random variable with rate $M := |E(B,r)| = \Theta(r^2)$. Using standard bounds for Poisson tail probabilities we get that $\Pr[N > {\e}^2 M T] = \exp(-\Omega(M T))$ for all $T \ge 1$. Therefore, 
		\[\Pr[\,\neg{\cal E}_\ell\,] \le \Pr[\,\neg{\cal E}_\ell\,|\,N \le {\e}^2 M T\,]+{\e}^{-\Omega(M T)}.\]
		
		Also, $\neg {\cal E}_\ell := \bigcup_{t:Q_t \le T} \,\neg {\cal E}_{\ell,t}$, and if the edge update at time $Q_t$ occurs outside $B$, we have $\neg{\cal E}_{\ell,t} = \neg{\cal E}_{\ell,t+1}$. Hence, a union bound implies
		\[\Pr[\,\neg{\cal E}_\ell\,] ~\le~ {\e}^2 M T \max_{t:Q_t \le T}  \Pr[\,\neg{\cal E}_{\ell,t}\,]+{\e}^{-\Omega(M T)}.\]
		
		Fact \ref{fact:exp-decay-extension} establishes exponential decay of finite volume connectivities (\ref{prelim:edc-finite}) for the distribution of $Y_t$ in $B$ for all $t \ge 0$. Then, as in Lemma \ref{lemma:propagation:dt}, we obtain
		\[\Pr[\,\neg{\cal E}_\ell\,] \le O(r^{6.25}) \cdot {\e}^{-\Omega(r^{1/4})} + {\e}^{-\Omega(r^2)}.\]
		Together with (\ref{propagation:first-term-ct}), this implies there exist constants $c,C,\Ll > 0$ such that for all $r \ge c$ we have $\Pr[\,X_{T}(e) \neq Y_{T}(e) \,] \le C \exp(-\lambda r^{1/4})$, as desired.  \end{proof}	
	
\end{document}